\documentclass[runningheads]{llncs}
\usepackage{amsmath,amssymb}
\usepackage{mathtools}
\usepackage{tikz}
\usepackage{ltl}
\usepackage{marvosym}
\usepackage[linesnumbered,lined,commentsnumbered,ruled,noend]{algorithm2e}

\usepackage{todonotes}

\usepackage{graphicx}                   
\usepackage{hyperref}                   
\usepackage[firstpage]{draftwatermark}  
\SetWatermarkAngle{0}

\SetWatermarkText{\raisebox{9.35cm}{%
  \hspace{0.1cm}%
  \href{https://doi.org/10.6084/m9.figshare.19690858.v8}{\includegraphics{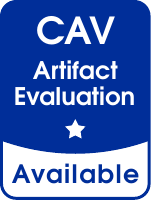}}%
  \hspace{9cm}%
  \includegraphics{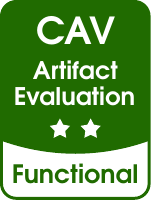}%
}}

\usepackage{color}
\usepackage{colortbl}
\usetikzlibrary{decorations.pathreplacing}
\usetikzlibrary{arrows,petri,backgrounds, positioning, shapes, patterns,calc,automata,intersections,calc}
\usepackage{wrapfig}
\usepackage{caption}
\usepackage{subcaption}
\usepackage{hyperref}
\usepackage{multirow}
\usepackage{scalerel}
\usepackage{cleveref}
\usepackage{array}
\usepackage{comment}
\usepackage{listings}

\usepackage{booktabs}
\newcommand{\true}[0]{\mathit{true}}
\newcommand{\false}[0]{\mathit{false}}

\newcommand{\mchyper}[0]{\textsc{MCHyper}\@\xspace}
\newcommand{\aiger}[0]{\textsc{Aiger}\@\xspace}
\newcommand{\abc}[0]{\textsc{ABC}\@\xspace}
\newcommand{\ldot}{\mathpunct{.}}
\newcommand{\U}{\LTLuntil}
\newcommand{\X}{\LTLnext}
\newcommand{\G}{\LTLglobally}
\newcommand{\F}{\LTLeventually}
\newcommand{\R}{\LTLrelease}

\newcommand{\Tr}{\mathit{Tr}}

\newcommand{\cex}{\Gamma}
\newcommand{\traces}{\mathit{traces}}
\newcommand{\intervene}{\mathit{intervene}}

\newcommand{\candidates}{\Tilde{\mathcal{C}}}
\newcommand{\cause}{\mathcal{C}}
\newcommand{\contingency}{\mathcal{W}}
\newcommand{\alt}{\mathcal{A}}
\newcommand{\traceassign}{\Pi}
\newcommand{\tracevars}{\mathcal{V}}
\renewcommand{\models}{\vDash}
\newcommand{\nmodels}{\nvDash}

\newcommand{\donotshow}[1]{}

\makeatletter
\def\moverlay{\mathpalette\mov@rlay}
\def\mov@rlay#1#2{\leavevmode\vtop{%
		\baselineskip\z@skip \lineskiplimit-\maxdimen
		\ialign{\hfil$\m@th#1##$\hfil\cr#2\crcr}}}
\newcommand{\charfusion}[3][\mathord]{
	#1{\ifx#1\mathop\vphantom{#2}\fi
		\mathpalette\mov@rlay{#2\cr#3}
	}
	\ifx#1\mathop\expandafter\displaylimits\fi}
\makeatother

\newcommand{\cupdot}{\charfusion[\mathbin]{\cup}{\cdot}}

\newcolumntype{H}{>{\setbox0=\hbox\bgroup}c<{\egroup}@{}}

\definecolor{TodoRed}{RGB}{225,63,63}
\newcommand{\ttodo}[4]{\ifthenelse{\equal{#1}{inline}}{\todo[inline, author=#2, color=#3,bordercolor=#3,backgroundcolor=#3!25,linecolor=#3]{#4}}{\todo[color=#3,bordercolor=#3,backgroundcolor=#3!25,linecolor=#3]{#2: #4}}}

\definecolor{trace0}{RGB}{179, 76, 18}
\definecolor{trace1}{RGB}{179, 152, 18}
\definecolor{expl0}{RGB}{141, 63, 144}
\definecolor{expl1}{RGB}{236, 77, 216}
\definecolor{highlight}{RGB}{0, 123, 255}

\newcommand{\aut}[1]{\ensuremath{\mathcal{#1}}}
\newcommand{\violation}[0]{\ensuremath{\Bar{\varphi}}}
\newcommand{\method}[1]{\ensuremath{\textsf{#1}}}

\newcommand{\cont}{{\method{ComputeContingency}}}
\newcommand{\ActualCause}{{\method{ActualCause}}}

\newcommand{\CF}{Counterfactual}

\newcommand{\hypervis}{HyperVis\xspace}

\begin{document}

\title{Explaining Hyperproperty Violations\thanks{This work was funded by DFG grant 389792660 as part of \href{https://perspicuous-computing.science}{TRR~248 -- CPEC}, by the DFG as part of the Germany’s Excellence Strategy EXC 2050/1 - Project ID 390696704  - Cluster of Excellence ``\emph{Centre for Tactile Internet}'' (CeTI) of TU Dresden, by the European Research Council (ERC)\ Grant OSARES (No. 683300), and by the German Israeli Foundation (GIF) Grant No. I-1513-407./2019.}}

\author{Norine Coenen\inst{1} \and Raimund Dachselt\inst{2} \and Bernd Finkbeiner\inst{1} \and Hadar Frenkel\inst{1} \and Christopher Hahn\inst{1} \and Tom Horak\inst{3} \and Niklas Metzger\inst{1} \and Julian Siber\inst{1}}
\institute{CISPA Helmholtz Center for Information Security, Saarbrücken, Germany\\ \email{\{norine.coenen,finkbeiner,hadar.frenkel,christopher.hahn, niklas.metzger,julian.siber\}@cispa.de} \and
	Interactive Media Lab, Technische Universit\"at Dresden, Dresden, Germany\\ \email{dachselt@acm.org} \and
	elevait GmbH \& Co. KG, Dresden, Germany\\
	\email{tom.horak@elevait.de}
}

\authorrunning{N.\ Coenen et al.}

\maketitle
\begin{abstract}
Hyperproperties relate multiple computation traces to each other.
Model checkers for hyperproperties thus return, in case a system model violates the specification, a set of traces as a counterexample. Fixing the erroneous relations between traces in the system that led to the counterexample is a difficult manual effort that highly benefits from additional explanations.
In this paper, we present an explanation method for counterexamples to hyperproperties described in the specification logic HyperLTL.
We extend Halpern and Pearl's definition of actual causality to sets of traces witnessing the violation of a HyperLTL formula, which allows us to identify the events that caused the violation. We report on the implementation of our method and show that it significantly improves on previous approaches for analyzing counterexamples returned by HyperLTL model checkers. 

\end{abstract}

\section{Introduction}

While model checking algorithms and tools~(e.g.,~\cite{DBLP:conf/lop/ClarkeE81,DBLP:journals/fmsd/ClarkeBRZ01,DBLP:conf/cav/DehnertJK017,10.1007/978-3-642-14295-6_5,DBLP:journals/sttt/LarsenPY97,DBLP:journals/tse/Holzmann97}) have, in the past, focused on trace properties,
recent failures in security-critical systems, such as Heartbleed~\cite{durumeric2014matter}, Meltdown~\cite{Lipp2018meltdown}, Spectre~\cite{Kocher2018spectre}, or Log4j~\cite{log4j}, have triggered the development of model checking algorithms for properties that relate multiple computation traces to each other, i.e., \emph{hyperproperties}~\cite{Hyperproperties}.
Although the counterexample returned by such a model checker for hyperproperties, which takes the shape of a \emph{set} of traces, may aid in the debugging process, understanding and narrowing down which features are actually responsible for the erroneous relation between the traces in the counterexample requires significantly more manual effort than for trace properties.
In this paper, we develop an explanation technique for these more complex counterexamples that identifies the \emph{actual causes}~\cite{Halpern15,HalpernPearl05a,HalpernPearl05b} of hyperproperty violations.

Existing hyperproperty model checking approaches (e.g.,~\cite{finkbeiner2018model,DBLP:conf/cav/FinkbeinerRS15,hsu2020bounded}), take a HyperLTL formula as an input. HyperLTL is a temporal logic extending LTL with explicit trace quantification~\cite{DBLP:conf/post/ClarksonFKMRS14}.

For example, observational determinism, which requires that all traces $\pi,\pi'$ agree on their observable outputs $\mathit{lo}$ whenever they agree on their observable inputs $\mathit{li}$, can be formalized in HyperLTL as $\forall \pi. \forall \pi'. \G (\mathit{li}_{\pi} \leftrightarrow \mathit{li}_{\pi'}) \rightarrow \G (\mathit{lo}_{\pi} \leftrightarrow \mathit{lo}_{\pi'}).$
In case a system model violates observational determinism, the model checker consequently returns a set of two execution traces witnessing the violation.

A first attempt in explaining model checking results of HyperLTL specifications has been made with \hypervis~\cite{DBLP:journals/tvcg/HorakCMHFMDFD22}, which visualizes a counterexample returned by the model checker MCHyper~\cite{DBLP:conf/cav/FinkbeinerRS15} in a browser application.
While the visualizations are already useful to analyze the counterexample at hand, it fails to identify causes for the violation in several security-critical scenarios.
 This is because \hypervis identifies important atomic propositions that appear in the HyperLTL formula and highlights these in the trace and the formula. For detecting causes, however, this is insufficient: a cause for a violation of observational determinism, for example, could be a branch on the valuation of a secret input $i_s$, which is not even part of the formula (see Sec.~\ref{sec:running_ex} for a running example).

Defining what constitutes an actual cause for an effect (a violation) in a given scenario is a precious contribution by Halpern and Pearl~\cite{Halpern15,HalpernPearl05a,HalpernPearl05b}, who refined and formalized earlier approaches based on counterfactual reasoning~\cite{Lewis1973}: Causes are sets of events such that, in the counterfactual world where they do not appear, the effect does not occur either. One of the main insights of Halpern and Pearl's work, however, is that naive counterfactuals are too imprecise. If, for instance, our actual cause preempted another potential cause, the mere absence of the actual cause will not be enough to prevent the effect, which will be still produced by the other cause in the new scenario. Halpern and Pearl's definition therefore allows to carefully control for other possible causes through the notion of \emph{contingencies}. In the modified definition~\cite{Halpern15}, contingencies allow to fix certain features of the counterfactual world to be exactly as they are in the actual world, regardless of the system at hand. Such a contingency effectively modifies the dynamics of the underlying model, and one insight of our work is that defining actual causality for reactive systems also needs to modify the system under a contingency. Notably, most works regarding trace causality~\cite{CaltaisGL18,GosslerM13} do not consider contingencies but only counterfactuals, and thus are not able to find true actual causes. 

In this paper, we develop the notion of actual causality for effects described by HyperLTL formulas and use the generated causes as explanations for counterexamples returned by a model checker. We show that an implementation of our algorithm is practically feasible and significantly increases the state-of-the-art in explaining and analyzing HyperLTL model checking results.

\section{Preliminaries}
We model a system as a \emph{Moore machine}~\cite{moore1956gedanken} $T = (S, s_0, AP, \delta, l)$ where 
$S$ is a finite set of states, $s_0 \in S$ is the initial state, $AP = I \cupdot O$ is the set of atomic propositions consisting of inputs $I$ and outputs $O$, $\delta: S \times 2^I \rightarrow S$ is the transition function determining the successor state for a given state and set of inputs, and $l: S \rightarrow 2^O$ is the labeling function mapping each state to a set of outputs. 
{A \emph{trace} $t = t_0 t_1 t_2 \ldots \in (2^{AP})^\omega$ of $T$ is an infinite sequence of sets of atomic propositions with $t_i = A \cup l(s_i)$, where $A \subseteq I$ and $\delta(s_{i}, A) = s_{i+1}$ for all $i \geq0$. We usually write $t[n]$ to refer to the set $t_n$ at the $n$-th position of $t$.}
With $\mathit{traces}(T)$, we denote the set of all traces of $T$. {For some sequence of inputs $a = a_0 a_1 a_2 \ldots \in (2^{I})^\omega$, the trace $T(a)$ is defined by $T(a)_i = a_i \cup l(s_i)$ and $\delta(s_{i}, a_i) = s_{i+1}$ for all $i \geq0$.}
A trace property $P \subseteq T$ is a set of traces. 
A hyperproperty $H$ is a lifting of a trace property, i.e., a \emph{set of sets of traces}.
A model $T$ satisfies a hyperproperty $H$ if the set of traces of $T$ is an element of the hyperproperty, i.e., $\mathit{traces}(T) \in H$.

\subsection{HyperLTL}\label{subsec:hyperltl}

HyperLTL is a recently introduced logic for expressing temporal hyperproperties, extending linear-time temporal logic (LTL)~\cite{DBLP:conf/focs/Pnueli77} with trace quantification:
\begin{align*}
\varphi &\Coloneqq \forall \pi \ldot \varphi \mid \exists \pi \ldot \varphi \mid~\psi~\\
\psi &\Coloneqq a_\pi \mid \neg \psi \mid \psi \land \psi \mid \X \psi \mid \psi \U \psi \enspace 
\end{align*}
We also consider the usual derived Boolean ($\lor$, $\rightarrow$, $\leftrightarrow$) and temporal operators ($\varphi \R \psi \equiv \neg(\neg \varphi \U \neg \psi)$, $\F \varphi \equiv \true \U \varphi$, $\G \varphi \equiv \false \LTLrelease \varphi$).
The semantics of HyperLTL formulas is defined with respect to a set of traces $\Tr$ and a trace assignment $\traceassign: \tracevars \rightarrow \Tr$ that maps trace variables to traces. 
To update the trace assignment so that it maps trace variable $\pi$ to trace $t$, we write $\traceassign[\pi \mapsto t]$. 
\begin{equation*}
\begin{array}{lll}
\traceassign,i \models_\Tr a_\pi       \qquad \qquad & \text{iff } & a \in \traceassign(\pi)[i] \\
\traceassign,i \models_\Tr \neg \varphi              & \text{iff } & \traceassign,i \nmodels_\Tr \varphi \\
\traceassign,i \models_\Tr \varphi \land \psi         & \text{iff } & \traceassign,i \models_\Tr \varphi \text{ and } \traceassign,i \models_\Tr \psi \\
\traceassign,i \models_\Tr \X \varphi                & \text{iff } & \traceassign,i+1 \models_\Tr \varphi \\
\traceassign,i \models_\Tr \varphi\U\psi             & \text{iff } & \exists j \geq i \ldot \traceassign, j \models_\Tr \psi \land \forall i \leq k < j \ldot \traceassign,k \models_\Tr \varphi \\
\traceassign,i \models_\Tr \exists \pi \ldot \varphi & \text{iff } & \text{there is some } t \in \Tr \text{ such that } \traceassign[\pi \mapsto t],i \models_\Tr \varphi\\
\traceassign,i \models_\Tr \forall \pi \ldot \varphi & \text{iff } & \text{for all } t \in \Tr \text{ it holds that } \traceassign[\pi \mapsto t],i \models_\Tr \varphi
\end{array}
\end{equation*}

We explain counterexamples found by \mchyper~\cite{DBLP:conf/cav/FinkbeinerRS15,DBLP:conf/cav/CoenenFST19}, which is a model checker for HyperLTL formulas, building on \abc~\cite{10.1007/978-3-642-14295-6_5}.
\mchyper takes as inputs a hardware circuit, specified in the \aiger format~\cite{Biere-FMV-TR-07-1}, and a HyperLTL formula.
\mchyper solves the model checking problem by computing the self-composition~\cite{DBLP:journals/mscs/BartheDR11} of the system.
If the system violates the HyperLTL formula, \mchyper returns a counterexample.
This counterexample is a set of traces through the original system that together violate the HyperLTL formula.
Depending on the type of violation, this counterexample can then be used to debug the circuit or refine the specification iteratively.

\subsection{Actual Causality}
\label{subsec:ac_causality}

A formal definition of what actually causes an observed effect in a given context has been proposed by Halpern and Pearl~\cite{HalpernPearl05a}. Here, we outline the version later modified by Halpern~\cite{Halpern15}. Causality is defined with respect to a \emph{causal model} $\mathcal{M} = (\mathcal{S},\mathcal{F})$, given by a \emph{signature} $\mathcal{S}$ and set of \emph{structural equations} $\mathcal{F}$, which define the dynamics of the system. A signature $\mathcal{S}$ is a tuple $(\mathcal{U},\mathcal{V},\mathcal{D})$, where
$\mathcal{U}$ and $\mathcal{V}$ are disjoint sets of variables, termed \emph{exogenous} and \emph{endogenous} variables, respectively; 
and $\mathcal{D}$ defines  the \emph{range} of possible values $\mathcal{D}(Y)$ for all variables $Y \in \mathcal{U}\cup\mathcal{V}$.
A \emph{context} $\vec{u}$ is an assignment to the variables in $\mathcal{U}\cup \mathcal{V}$ such that 
the values of the exogenous variables are determined by factors outside of the model, while the value of some endogenous variable $X$ is defined by the associated structural equation $f_X \in \mathcal{F}$. 
 An \emph{effect} $\varphi$ in a causal model is a Boolean formula over assignments to endogenous variables. 
We say that a context $\vec{u}$ of a model $\mathcal{M}$ satisfies a partial variable assignment $\vec{X} = \vec{x}$ for $\vec{X}\subseteq \mathcal{U}\cup \mathcal{V}$ if the assignments in $\vec{u}$ and in $\vec{x}$ coincide for every variable $X \in \vec{X}$. The extension for Boolean formulas over variable assignments is as expected. 
For a context $\vec{u}$ and a partial variable assignment $\vec{X} = \vec{x}$, we denote by $(\mathcal{M}, \vec{u})[\vec{X} \leftarrow \vec{x}]$ the context $\vec{u}'$ in which the values of the variables in $\vec{X}$ are set according to $\vec{x}$, and all other values are computed according to the structural equations. 

The actual causality framework of Halpern and Pearl aims at defining what events (given as variable assignments) are the cause for the occurrence of an effect in a specific given context. We now provide the formal definition.

\begin{definition}[\cite{HalpernPearl05a,Halpern15}]\label{def:HP_causality}
A partial variable assignment	$\vec{X} = \vec{x}$ is an \emph{actual cause} of the effect $\varphi$ in $(\mathcal{M},\vec{u})$ if the following three conditions hold.
	\begin{description}
		\item AC1: $(\mathcal{M},\vec{u}) \models \vec{X} = \vec{x}$ and $(\mathcal{M},\vec{u}) \models \varphi$, i.e., both cause and effect are true in the actual world.
		\item AC2: There is a set $\vec{W}\subseteq V$ of endogenous variables and an assignment $\vec{x}'$ to the variables in $\vec{X}$ s.t. if $(\mathcal{M},\vec{u}) \models \vec{W} = \vec{w}$, then $(\mathcal{M},\vec{u})  [\vec{X} \leftarrow \vec{x}',\vec{W} \leftarrow \vec{w} ] \models \lnot \varphi$.
		\item AC3: $\vec{X}$ is minimal, i.e. no subset of $\vec{X}$ satisfies AC1 and AC2.
	\end{description}
\end{definition}

Intuitively, AC2 states that in the counterfactual world obtained by intervening on the cause $\vec{X} = \vec{x}$ in the actual world (that is, setting the variables in $\vec{X}$ to $\vec{x}'$), the effect does not appear either. However, intervening on the possible cause might not be enough, for example when that cause preempted another. After intervention, this other cause may produce the effect again, therefore clouding the effect of the intervention. {To address this problem, AC2 allows to reset values through the notion of \emph{contingencies}, i.e., the set of variables $\vec{W}$ can be reset to $\vec{w}$, which is (implicitly) universally quantified. However, since the actual world has to model $\vec{W} = \vec{w}$, it is in fact uniquely determined. AC3, lastly, enforces the cause to be minimal by requiring that all variables in $\vec{X}$ are strictly necessary to achieve AC1 and AC2.} For an illustration of Halpern and Pearl's actual causality, see Ex.~\ref{ex:hp_causality} in Sec.~\ref{sec:running_ex}.

\section{Running Example}\label{sec:running_ex}

Consider a security-critical setting with two security levels: a high-security level $h$ and a low-security level $l$.
Inputs and outputs labeled as high-security, denoted by $\mathit{hi}$ and $\mathit{ho}$ respectively, are confidential and thus only visible to the user itself, or, e.g., admins.
Inputs and outputs labeled as low-security, denoted by $\mathit{li}$ and $\mathit{lo}$ respectively, are public and are considered to be observable by an attacker.
\begin{wrapfigure}{r}{0.3\textwidth}
    \centering
    \resizebox{.3\textwidth}{!}{
    \begin{tikzpicture}
			[->,shorten >=1pt,auto,node distance=3cm and 4cm, on grid,initial text=,
			every state/.style={minimum size=45pt,inner sep=0pt},
			every node/.style={font=\small}]
			
			\node(s0) [state, initial above] {\begin{tabular}{c}
					$s_0$\\ $\emptyset$
				\end{tabular}};
			\node(s1) [state, right of=s0,yshift=2.5em] {\begin{tabular}{c}
					$s_1$\\ $\{\mathit{ho}\}$
			\end{tabular}};
		\node(s2) [state, below of=s0] {\begin{tabular}{c}
		$s_2$\\ $\{\mathit{lo}\}$
	\end{tabular}};
\node(s3) [state, right of=s2,yshift=2.5em] {\begin{tabular}{c}
$s_3$\\ $\{\mathit{ho},\mathit{lo}\}$
\end{tabular}};
			
			\path[->] (s0)	edge node[pos=0.55]{$\mathit{hi}$}  (s1)
			(s0)	edge node[swap]{$\lnot \mathit{hi}$}  (s2)
			(s2)	edge node[swap]{$\mathit{hi}$}  (s1)
			(s1)	edge node[]{$\top$}  (s3)
			(s2)	edge node[pos=0.25,swap]{$\lnot\mathit{hi}$}  (s3)
			(s3)	edge[loop below] node[]{$\top$}  (s3);
		\end{tikzpicture}
		}
    \caption{State graph representation of our example system.}
    \label{fig:running_example_system}
\end{wrapfigure}
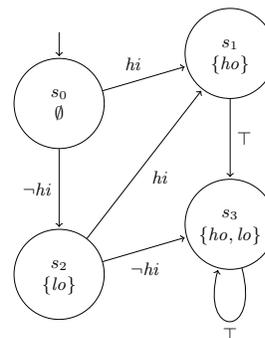
Our system of interest is modeled by the state graph representation shown in Fig. \ref{fig:running_example_system}, which is treated as a black box by an attacker.
The system is run without any low-security inputs, but branches depending on the given high-security inputs. If in one of the first two steps of an execution, a high-security input $\mathit{hi}$ is encountered, the system outputs only the high-security variable $\mathit{ho}$ directly afterwards and in the subsequent steps both outputs, regardless of inputs. If no high-security input is given in the first step, the low-security output $\mathit{lo}$ is enabled and after the second step, again both outputs are enabled, regardless of what input is fed into the system.

A prominent example hyperproperty is \emph{observational determinism} from the introduction:
which states that any sequence of low-inputs always produces the same low-outputs, regardless of what the high-security level inputs are.
$
\varphi = \forall \pi. \forall \pi'. \G (\mathit{li}_{\pi} \leftrightarrow \mathit{li}_{\pi'}) \rightarrow \G (\mathit{lo}_{\pi} \leftrightarrow \mathit{lo}_{\pi'})$.
The formula states that all traces $\pi$ and $\pi'$ must agree in the low-security outputs if they agree in the low-security inputs. Our system at hand does not satisfy observational determinism, because the low-security outputs in the first two steps depend on the present high-security inputs. Running MCHyper, a model checker for HyperLTL, results in the following counterexample:
$ t_1 =  \{\} \{\mathit{lo}\} \{\mathit{ho},\mathit{lo}\}^\omega$ and 
 $t_2 = \{\mathit{hi}\} \{\mathit{hi},\mathit{ho}\} \{\mathit{ho},\mathit{lo}\}^\omega$ .
With the same low-security input (none) the traces produce different low-security outputs by visiting $s_1$ or $s_2$ on the way to $s_3$.

In this paper, our goal is to explain the violation of a HyperLTL formula on such a counterexample. Following Halpern and Pearl's explanation framework~\cite{HalpernPearl05b}, an actual cause that is considered to be possibly true or possibly false constitutes an explanation for the user. We only consider causes over input variables, which can be true and false in any model. Hence, finding an explanation amounts to answering which inputs caused the violation on a specific counterexample. Before we answer this question for HyperLTL and the corresponding counterexamples given by sets of traces (see Sec.~\ref{sec:trace_causality}), we first illustrate Halpern and Pearl's actual causality (see Sec.~\ref{subsec:ac_causality}) with the above running example.

\begin{example}\label{ex:hp_causality}
Finite executions of a system can be modeled in Halpern and Pearl's causal models. Consider inputs as exogenous variables $\mathcal{U} = \{\mathit{hi}_0,\mathit{hi}_1\}$ and outputs as endogenous variables $\mathcal{V} = \{\mathit{lo}_1,\mathit{lo}_2,\mathit{ho}_1,\mathit{ho}_2\}$. The indices model at which step of the execution the variable appears. We omit the inputs at the third position and the outputs at the first position because they are not relevant for the following exposition. We have that $\mathcal{D}(Y) = \{0,1\}$ for every $Y \in \mathcal{U}\cup\mathcal{V}$. Now, the following manually constructed structural equations encode the transitions: (1) $\mathit{lo}_1 = \lnot \mathit{hi}_0$, (2) $\mathit{ho}_1 = \mathit{hi}_0$, (3) $\mathit{lo}_2 = \lnot \mathit{hi}_1 \lor \lnot \mathit{lo}_1$ and (4) $\mathit{ho}_2 =  \mathit{lo}_1 \lor \mathit{ho}_1$. Consider context $\vec{u} = \{\mathit{hi}_0 = 0, \mathit{hi}_1 = 1\}$, effect $\varphi \equiv \mathit{lo}_1 = 1 \lor \mathit{lo}_2 = 1$, and candidate cause $\mathit{hi}_0 = 0$. Because of Eq.~(1), we have that $(\mathcal{M},\vec{u}) \models {hi}_0 = 0$ and $(\mathcal{M},\vec{u}) \models \mathit{lo}_1 = 1$, hence AC1 is satisfied. Regarding AC2, this example allows us to illustrate the need for contingencies to accurately determine the actual cause: If we only consider intervening on the candidate cause $\mathit{hi}_0 = 0$, we still have $(\mathcal{M},\vec{u})[\mathit{hi}_0 \leftarrow 1] \models \varphi$, because with $\mathit{lo}_1 = 0$ and Eq.~(3) it follows that $(\mathcal{M},\vec{u}) \models \mathit{lo}_2 = 1$. However, in the actual world, the second high input has no influence on the effect. We can control for this by considering the contingency $\mathit{lo}_2 = 0$, which is satisfied in the actual world, but not after the intervention on $\mathit{hi}_0$. Because of this contingency, we then have that $(\mathcal{M},\vec{u})[\mathit{hi}_0 \leftarrow 1, \mathit{lo}_2 \leftarrow 0] \models \lnot \varphi$, and hence, AC2 holds. Because a singleton set automatically satisfies AC3, we can infer that the first high input $\mathit{hi}_0$ was the actual cause for any low output to be enabled in the actual world. Note that, intuitively, the contingency allows us to ignore some of the structural equations by ignoring the value they assign to $\mathit{lo}_2$ in this context. Our definitions in Sec.~\ref{sec:trace_causality} will allow similar modifications for counterexamples to hyperproperties.
\end{example}

\section{Causality for Hyperproperty Violations}\label{sec:trace_causality}

Our goal in this section is to formally define actual causality for the violation of a hyperproperty described by a general HyperLTL formula $\varphi$, observed in a counterexample to $\varphi$. Such a counterexample is given by a trace assignment to the trace variables appearing in $\varphi$. Note that, for universal quantifiers, the assignment of a single trace to the bounded variable suffices to define a counterexample. For existential quantifiers, this is not the case: to prove that an existential quantifier cannot be instantiated we need to show that no system trace satisfies the formula in its body, i.e., provide a proof for the whole system. In this work, we are interested in explaining violations of hyperproperties, and not proofs of their satisfaction~\cite{DBLP:journals/tocl/ChocklerHK08}. Hence, we limit ourselves to instantiations of the outermost universal quantifiers of a HyperLTL formula, which can be returned by model checkers like MCHyper~\cite{DBLP:conf/cav/FinkbeinerRS15,DBLP:conf/cav/CoenenFST19}.
Since our goal is to explain counterexamples, restricting ourselves to results returned by existing model checkers is reasonable. 
Note that MCHyper can still handle formulas of the form $\forall^n \exists^m \varphi$ where $\varphi$ is quantifier free, including interesting information flow policies like generalized noninterference~\cite{GNI}.
The returned counterexample then only contains $n$ traces that instantiate the universal quantifiers, the existential quantifiers are not instantiated for the above reason. 
In the following, we restrict ourselves to formulas and counterexamples of this form. 

\begin{definition}[Counterexample]\label{def:universal_cex}
Let $T$ be a transition system and denote $\mathit{Traces}(T) := \Tr$, and let $\varphi$ be a HyperLTL formula of the form $\forall \pi_1 \ldots \forall \pi_k \psi$, where $\psi$ is a HyperLTL formula that does not start with $\forall$. A counterexample to $\varphi$ in $T$ is a partial trace assignment $\cex: \{\pi_1,\ldots,\pi_k\} \rightarrow \Tr$ such that $\cex,0 \models_{\Tr} \lnot \psi$. 
\end{definition}

For ease of notation, we sometimes refer to $\cex$ simply as the tuple of its instantiations $\cex = \langle \cex(\pi_1), \ldots, \cex(\pi_k) \rangle $. 
In terms of Halpern and Pearl's actual causality as outlined in Sec.~\ref{subsec:ac_causality}, a counterexample describes the actual world at hand, which we want to explain. As a next step, we need to define an appropriate language to reason about possible causes and contingencies in our counterexample. We will use sets of \emph{events}, i.e., values of atomic propositions at a specific position of a specific trace in the counterexample.

\begin{definition}[Event]
An event is a tuple $ e = \langle l_a, n, t \rangle$ such that $l_a = a$ or $l_a = \neg a$ for some atomic proposition
$a \in AP$, $n\in \mathbb{N}$ is a point in time, and $t \in (2^{AP})^\omega$ is a trace of a system $T$. We say that a counterexample $\cex = \langle t_1, \ldots t_k \rangle$ \emph{satisfies}
a set of events $\cause$, and denote $\cex \models \cause$, 
if for every event $ \langle l_a, n, t \rangle \in \cause$
the two following conditions hold:
\begin{enumerate}
    \item  $t = t_i$ for some $i\in \{1,\ldots,k\}$, i.e., all events in $\cause$ reason about traces in $\cex$,
    \item  {$l_a = a$ iff $a\in t_i[n]$, i.e.,
$a$ holds on trace $t_i$ of the counterexample at time $n$.}  
\end{enumerate}
\end{definition}

We assume that the set $\mathit{AP}$ is a disjoint union of input an output propositions, that is, $\mathit{AP} = I \cupdot O $. 
We say that $\langle l_a,n,t \rangle$ is an \emph{input event} if $a\in I$, and we call it an \emph{output event} if $a\in O$. We denote the set of input events by $\mathit{IE}$ and the set of output events by $\mathit{OE}$. These events have a direct correspondence with the variables appearing in Halpern and Pearl's causal models: we can identify input events with exogenous variables (because their value is determined by factors outside of the system) and output events with endogenous variables.

We define a cause as a set of input events, while
 an effect is a possibly infinite Boolean formula over $OE$. Note that, similar to~\cite{conf/stacs/Finkbeiner017}, every HyperLTL formula can be represented as a first order formula over events, e.g. $\forall\pi\forall\pi'\G (a_\pi \leftrightarrow a_{\pi'}) =  \forall\pi\forall\pi'\bigwedge_{n\in\mathbb{N}}  (\langle a, n, \pi \rangle \leftrightarrow \langle a, n, \pi' \rangle)$. For some set of events $\mathcal{S}$, let ${}^+\!\mathcal{S}_\pi^k = \{ a \in \mathit{AP}~|~\langle a, k, \pi \rangle \in \mathcal{S}\}$ denote the set of atomic propositions defined positively by $\mathcal{S}$ on trace $\pi$ at position $k$. Dualy, we define ${}^-\!\mathcal{S}_\pi^k = \{ a \in \mathit{AP}~|~\langle \lnot a, k, \pi \rangle \in \mathcal{S}\}$.

In order to define actual causality for hyperproperties we need to formally define how we obtain the counterfactual executions under some contingency for the case of events on infinite traces. We define a contingency as a set of output events. Mapping Halpern and Pearl's definition to transition systems, contingencies reset outputs in the counterfactual traces back to their value in the original counterexample, which amounts to changing the state of the system, and then following the transition function from the new state. For a given trace of the counterexample, we describe all possible behaviors under \emph{arbitrary} contingencies with the help of a counterfactual automaton. The concrete contingency on a trace is defined by additional input variables. In the following, let $\mathit{IC} = \{o^C~|~o \in O \}$ be a set of auxiliary input variables expressing whether a contingency is invoked at the given step of the execution and $c : O \to \mathit{IC}$ be a function s.t. $c(o) = o^C$. 

\begin{definition}[\CF~Automaton]\label{def:CF_automaton}
Let $T = (S, s_0, \mathit{AP}, \delta, l)$ be a system with $S = 2^{\mathit{O}}$, i.e., every state is uniquely labeled, and there exists a state for every combination of outputs. Let $\pi = \pi_0 \ldots \pi_i (\pi_{j} \ldots \pi_{n})^\omega \in \traces(T)$ be a trace of $T$ in a finite, lasso-shaped representation. The counterfactual automaton $T^C_\pi = (S \times \{0 \ldots n\}, (s_0,0), (\mathit{IC} \cupdot I) \cupdot (O \cupdot \{0 \ldots n\}), \delta^C, l^C)$ is defined as follows:
\begin{itemize}
    \item $\delta^C((s,k),Y) = (s',k')$ where $k'=j$ if $k = n$, else $k' = k+1$, and\\ $l(s') = \{ o \in O~|~(o \in \delta(s,Y\cap I) \land c(o) \not\in Y) \lor (o \in \pi_{k'} \land c(o) \in Y)\}$,
    \item $l^C(s,k) = l(s) \cup \{k\}$ .
\end{itemize}

\end{definition}

A counterfactual automaton is effectively a chain of copies of the original system, of the same length as the counterexample. An execution through the counterfactual automaton starts in the first copy corresponding to the first position in the counterexample trace, and then moves through the chain until it eventually loops back from copy $n$ to copy $j$. A transition in the counterfactual automaton can additionally specify setting as a contingency some output variable $o$ if the auxiliary input variable $o^C$ is enabled. In this case, the execution will move to a state in the next automaton of the chain where all the outputs are as usual, except $o$, which will have the same value as in the counterexample $\pi$. Note that, under the assumption that all states of the original system are uniquely labeled and there exists a state for every combination of output variables, the function $\delta^C$ is uniquely determined. \footnote{The same reasoning can be applied to arbitrary systems by considering for contingencies largest sets of outputs for which the assumption holds, with the caveat that the counterfactual automaton may model fewer contingencies. Consequently, computed causes may be less precise in case multiple causes appear in the counterexample.} A counterfactual automaton for our running example is described in App.~\ref{app:counterfactual_automaton}.

Next, we need to define how we intervene on a set of traces with a candidate cause given as a set of input events, and a contingency given as a set of output events. We define an intervention function, which transforms a trace of our original automaton to an input sequence of an counterfactual automaton.

\begin{definition}[Intervention]
For a cause $\cause \subseteq \mathit{IE}$, a contingency $\contingency \subseteq \mathit{OE}$ and a trace $\pi$ , the function $\intervene : (2^\mathit{AP})^\omega \times 2^\mathit{IE} \times 2^\mathit{OE} \to (2^{\mathit{I}\cup \mathit{IC}})^\omega$ returns a trace such that for all $k \in \mathbb{N}$ the following holds:
$ \intervene(\pi,\cause,\contingency)[k] = (\pi[k] \setminus {}^+\cause_\pi^k) \cup {}^-\cause_\pi^k \cup \{ c(o)~|~o \in \!{}^+\contingency_\pi^k \cup \!{}^-\contingency_\pi^k\}.$
We lift the intervention function to counterexamples given as a tuple $\cex = \langle \pi_1, \ldots, \pi_k \rangle$ as follows: 
$\intervene(\cex,\cause,\contingency) =\\ \langle T^C_{\pi_1}(\intervene(\pi_1,\cause,\contingency)), \ldots , T^C_{\pi_k}(\intervene(\pi_k,\cause,\contingency))\rangle.$

\end{definition}

Intuitively, the intervention function \emph{flips} all the events that appear in the cause $\cex$: If some $a \in \mathit{I}$ appears positively in the candidate cause $\cause$, it will appear negatively in the resulting input sequence, and vice-versa. For a contingency $\contingency$, the intervention function enables their auxiliary input for the counterfactual automaton at the appropriate time point irrespective of their value, as the counterfactual automaton will take care of matching the atomic propositions value to the value in the original counterexample $\cex$.

\subsection{Actual Causality for HyperLTL Violations}

We are now ready to formalize what constitutes an actual cause for the violation of a hyperproperty described by a HyperLTL formula.

\begin{definition}[Actual Causality for HyperLTL] \label{def:hyper_causality}
Let $\cex$ be a counterexample to a HyperLTL formula $\varphi$ in a system $T$.
The set $\cause$ is an actual cause for the violation of $\varphi$ on $\cex$ if the following conditions hold. 
	\begin{description}
	\item[SAT] $\cex \models \cause$.
		\item[CF] There exists a contingency $\contingency$ and a non-empty subset $\cause' \subseteq \cause$ such that:
		$
		\cex \models \contingency \text{  and  } \intervene(\cex,\cause',\contingency)  \models_{\traces(T)} \varphi
		$.
		\item[MIN] $\cause$ is minimal, i.e., no subset of $\cause$ satisfies SAT and CF.
	\end{description}
\end{definition}

Unlike in Halpern and Pearl's definition (see Sec.~\ref{subsec:ac_causality}), the condition SAT requires $\cex$ to satisfy only the cause, as
we already know that the effect $\lnot \varphi$, i.e., the violation of the specification, is satisfied by virtue of $\cex$ being a counterexample. 
CF is the counterfactual condition corresponding to AC2 in Halpern and Pearl's definition, and it states that after intervening on the cause, under a certain contingency, the set of traces satisfies the property.
(Note that we use a conjunction of two statements here while Halpern and Pearl use an implication. This is because they implicitly quantify universally over the values of the variables in the set $W$ (which should be as in the actual world) where in our setting the set of contingencies already defines explicit values.) 
MIN is the minimality criterion directly corresponding to AC3.

\begin{example}\label{ex:hyper_causality}
Consider our running example from Sec.~\ref{sec:running_ex}, i.e., the system from Fig.~\ref{fig:running_example_system} and the counterexample to observational determinism $\cex = \langle t_1, t_2 \rangle$. Let us consider what it means to intervene on the cause $\cause_1 = \{\langle\mathit{hi},0,t_2\rangle\}$. Note that we have $\cex \models \cause_1$, hence the condition SAT is satisfied. For CF, let us first consider an intervention without contingencies. This results in $\intervene(\cex,\cause_1,\emptyset) = \langle t_1', t_2'\rangle = \langle t_1, \{\} \{\mathit{hi},\mathit{lo}\} \{\mathit{ho}\} \{\mathit{ho},\mathit{lo}\}^\omega\rangle$.
However, $\intervene(\cex,\cause_1,\emptyset) \models_{\traces(T)} \lnot \varphi$, because the low outputs of $t_1'$ and $t_2'$ differ at the third position: $ \mathit{lo} \in t_1'[2]$ and $ \mathit{lo} \not\in t_2'[2]$. This is because now the second high input takes effect, which was preempted by the first cause in the actual counterexample. 
The contingency $\contingency_2 = \{\langle\mathit{lo},2,t_2 \rangle\rangle\}$ now allows us to control this by \emph{modyfing the state} after taking the second high input as~follows: $
\intervene(\cex,\cause_2,\contingency_2)) = \langle t_1'' , t_2'' \rangle = \langle t_1, \{\} \{\mathit{hi},\mathit{lo}\} \{\mathit{ho},\mathit{lo}\} \{\mathit{ho},\mathit{lo}\}^\omega\rangle$.
Note that $t_2''$ is not a trace of the model depicted in Fig.~\ref{fig:running_example_system}, because there is no transition that explains the step from $t_2''[1]$ to $t_2''[2]$. It is, however, a trace of the counterfactual automaton $T^C_{t_2}$ (see App.~\ref{app:counterfactual_automaton}), which encodes the set of counterfactual worlds for the trace $t_2$. The fact that we consider executions that are not part of the original system allows us to infer that only the first high input was an actual cause in our running example. Disregarding contingencies, we would need to consider both high inputs as an explanation for the violation of observational determinism, even though the second high input had no influence. Our treatment of contingencies corresponds directly to Halpern and Pearl's causal models, which allow to ignore certain structural equations as outlined in Ex.~\ref{ex:hp_causality}.
\end{example}

\emph{Remark:} 
With our definitions, we strictly generalize Halpern and Pearl's actual causality to reactive systems modeled as Moore machines and effects expressed as HyperLTL formulas. 
Their structural equation models can be encoded in a one-step Moore machine; effect specifying a Boolean combination of primitive events can be encoded in the more expressive logic HyperLTL. 
Just like for Halpern and Pearl, our actual causes are not unique. 
While there can exist several different actual causes, the set of all actual causes is always unique. 
It is also possible that no actual cause exists: If the effect occurs on all system traces, there may be no actual cause on a given individual trace.

\subsection{Finding Actual Causes with Model Checking }
In this section, we consider the relationship between finding an actual cause for the violation of a HyperLTL formula starting with a universal quantifier and model checking of HyperLTL. 
We show that the problem of finding an actual cause can be reduced to a model checking problem where the generated formula for the model checking problem has one additional quantifier alternation. 
While there might be a reduction resulting in a more efficient encoding, our current result suggests that causality checking is the harder problem. 
The key idea of our reduction is to use counterfactual automata (that encode the given counterexample and the possible counterfactual traces) together with the HyperLTL formula described in the proof to ensure the conditions SAT, CF, and MIN on the witnesses for the model checking result. 

\begin{proposition}\label{cc_as_mc}
We can reduce the problem of finding an actual cause for the violation of an HyperLTL formula starting with a universal quantifier to the HyperLTL model checking problem with one additional quantifier alternation. 
\end{proposition}

\begin{proof}
Let $\cex = \langle t_1, \ldots t_k \rangle$ be a counterexample for the formula $\forall \pi_1 \ldots \forall \pi_k .\varphi$ where $\varphi$ is a HyperLTL formula that does not have a universal first quantifier. 
We provide the proof for the case of $\cex = \langle t_1, t_2 \rangle$ for readability reasons, but it can be extended to any natural number $k$. 
We assume that 
$t_1, t_2$ have some $\omega$-regular representation, as otherwise the initial problem of computing causality is not well defined.
That is, we denote $t_i = u_i (v_i)^\omega$ such that $|u_i \cdot v_i| = n_i$.

In order to find an actual cause, we need to find a pair of traces $t_1', t_2'$ that are counterfactuals for $t_1, t_2$; satisfy the property $\varphi$; and the changes from $t_1, t_2$ to $t_1', t_2'$  are minimal with respect to set containment. Changes in inputs between $t_i$ and $t_i'$ in the loop part $v_i$ should reoccur in $t_i'$ repeatedly. 
Note that the differences between the counterexample $\langle t_1, t_2\rangle$ and the witness of the model checking problem $\langle t_1', t_2'\rangle$ encode the actual cause, i.e. in case of a difference, the cause contains the event that is present on the counterexample. 
To reason about these changes, we use the counterfactual automaton $T_i^C$ for each $t_i$, {which also allows us to search for the contingency $\mathcal{W}$ as part of the input sequence of $T_i^C$.} Note that each $T_i^C$ consists of $n_i$ copies, that indicate in which step the automaton is with respect to $t_i$ and its loop $v_i$. 
For $m>|u_i|$, we label each state $(s_i, m)$ in $T_i^C$  with the additional label 
$L_{s_m, i}$, to indicate that the system is now in the loop part of $t_i$. In addition, we add to the initial state of $T_i^C$ the label $l_i$, and we add to the initial state of the system $T$ the label $l_{\mathit{or}}$.  
 The formula $\psi^i_{\mathit{loop}}$ below states that the
 trace $\pi$ begins its run from the initial state of $T_i^C$ (and thus stays in this component through the whole run), and that
 every time $\pi$ visits a state on the loop, the same input sequence is observed. This way we enforce the periodic input behavior of the traces $t_1, t_2$ on $t_1', t_2'$. 
$$ \psi^i_{\mathit{loop}}(\pi) := 
l_{i, \pi} \wedge 
\bigwedge_{L_{s_m,i}} \bigvee_{A\subseteq I} \G (L_{{s_m,i}, \pi} \rightarrow 
(\bigwedge_{a\in A} a_\pi \wedge \bigwedge_{a\notin A} \neg a_\pi))
$$

For a subset of locations $N\subseteq [1,n_i]$ and a subset of input propositions $A\subseteq I$ we define $\psi^i_{\mathit{diff}}[N, A] (\pi)$ that states that $\pi$ differs from $t_i$ in at least all events $\langle l_a, m, t_i \rangle $ for $a\in A, m\in N$; and the formula $\psi^i_{\mathit{eq}}[N, A] (\pi)$ that states that
for all events that are not defined by $A$ and $N$, 
$\pi$ 
is equal to $t_i$.

$$
\psi^i_{\mathit{diff}}[N, A] (\pi) = \bigwedge_{j\in N , a\in A}\X^j (a_\pi \not\leftrightarrow a_{t_i})   
$$

$$
\psi^i_{\mathit{eq}}[N, A] (\pi) = \bigwedge_{j\notin N , a\in I}\X^j (a_\pi \leftrightarrow a_{t_i})   \wedge \bigwedge_{j\in [1,n_i] , a\notin A}\X^j (a_\pi \leftrightarrow a_{t_i})   
$$

We now define the formula $\psi^i_{\mathit{min}}$ that states that the set of inputs (and locations) on which trace $\pi$
differs from $t_i$ is not contained in the corresponding set for $\pi'$. We only check locations up until the length $n_i$ of $t_i$.

$$
\psi^i_{\mathit{min}} (\pi, \pi') :=  \bigwedge_{N \subseteq[i, n_i]}  \bigwedge_{A\subseteq\aut{I}}   
\left(\left(\psi^i_{\mathit{diff}}[N, A] (\pi) \wedge  
\psi^i_{\mathit{eq}}[N, A] (\pi) \right)
\rightarrow \neg \psi^i_{\mathit{eq}}[N, A] (\pi') \right)
$$

Denote $\varphi := Q_1\tau_1 \ldots Q_n \tau_n. ~ \varphi' (\pi_1, \pi_2)$ where $Q_i\in \{ \forall, \exists\}$ and $\tau_i$ are trace variables for $i\in[1,n]$. The formula
$\psi_\mathit{cause}$ described below 
states that the two traces $\pi'_1$ and $\pi'_2$ are part of the systems $T_1^C, T_2^C$, and have the same loop structure as $t_1$ and $t_2$, and satisfy $\varphi$. That is, these traces can be obtained by changing the original traces $t_1, t_2$ and avoid the violation.  

$$
  \psi_\mathit{cause}(\pi'_1, \pi'_2) := \varphi'(\pi'_1,  \pi'_2) \wedge \bigwedge_{i=1,2} \psi^i_{\mathit{loop}}(\pi'_i)
$$

Finally, $\psi_\mathit{actual}$ described below 
states that
the counterfactuals $\pi_1', \pi_2'$ correspond to a minimal change in the input events with respect to $t_1, t_2$. 
All other traces that the formula reasons about start at the initial state of the original system and thus are not affected by the counterfactual changes. We verify $\psi_\mathit{actual}$
 against the product automaton $ T \times T_1^C \times T_2^C$ to find these traces $\pi_i' \in T^C_i$ that witness the presence of a cause, counterfactual and contingency.

\begin{align*}
  \psi_\mathit{actual} := \exists \pi'_1  \exists \pi'_2 . ~
 \forall \pi_1'' \pi_2'' . ~ Q_1\tau_1 \ldots Q_n \tau_n . ~
 \psi_\mathit{cause}(\pi'_1, \pi'_2) 
 \wedge \bigwedge_{i=1,2} (l_{i, \pi'_i} \wedge l_{i, \pi''_i}) \\
 \wedge \bigwedge_{i\in[1,n]} l_{\mathit{or}, \tau_i}  
 \wedge 
\left(  \psi_\mathit{cause}(\pi_1'', \pi_2'')  \rightarrow \left( 
 \bigwedge_{i=1,2} \psi^i_{\mathit{min}}(\pi'_i,\pi''_i) \right) \right)  
\end{align*}

Then, if there exists two such traces $\pi_1', \pi_2'$ in the system $ T \times T_1^C \times T_2^C$, they correspond to a minimal cause for the violation. Otherwise, there are no traces of the counterfactual automata that can be obtained from $t_1, t_2$ using counterfactual reasoning and satisfy the formula $\varphi$. 
\qed
\end{proof}

We have shown that we can use HyperLTL  model checking to find an actual cause for the violation of a HyperLTL formula. 
The resulting model checking problem has an additional quantifier alternation which suggests that identifying actual causes is a harder problem. 
Therefore, we restrict ourselves to finding actual causes for violations of universal HyperLTL formulas. 
This keeps the algorithms we present in the next section practical 
as we start without any quantifier alternation and need to solve a model checking problem with a single quantifier alternation. 
While this restriction excludes some interesting formulas, many can be strengthened into this fragment such that we are able to handle close approximations (c.f.~\cite{softwaredoping}). 
Any additional quantifier alternation from the original formula carries over to an additional quantifier alternation in the resulting model checking problem which in turn leads to an exponential blow-up. 
The scalability of our approach is thus limited by the complexity of the model checking problem. 

\section{Computing Causes for Counterexamples}
\label{sec:computing-semantic-explanations}

In this section, we describe our algorithm for finding actual causes of hyperproperty violations. 
Our algorithm is implemented on top of MCHyper~\cite{DBLP:conf/cav/FinkbeinerRS15}, a model checker for hardware circuits and the alternation-free fragment of HyperLTL. 
In case of a violation, our analysis enriches the provided counterexample with the actual cause which can explain the reason for the violaiton to the user.

We first provide an overview of our algorithm and then discuss each step in detail. 
First, we compute an over-approximation of the cause using a satisfiability analysis over the transitions taken in the counterexample.
This analysis results in a set of candidate events $\candidates$. 
As we show in Prop.~\ref{prop:overapp}, every actual cause $\cause$ for the violation is a subset of $\candidates$. 
In addition, in Prop.~\ref{prop:CF} we show that the set $\candidates$ satisfies conditions SAT and CF. 
To ensure MIN, we search for the smallest subset $\cause \subseteq \candidates$ that satisfies SAT and CF. 
This set $\cause$ is then our minimal and therefore actual cause. 

To check condition CF, we need to check the counterfactual of each candidate cause $\cause$, and potentially also look for contingencies for $\cause$. 
We separate our discussion as follows. We first discuss the calculation of the over-approximation $\candidates$ (Sec.~\ref{sec:candidates}), then we present the $\ActualCause$ algorithm that identifies a minimal subset of $\candidates$ that is an actual cause (Sec.~\ref{sec:checking_causality}), and finally we discuss in detail the calculation of contingencies (Sec.~\ref{sec:contingencies}).

In the following sections, we use a reduction of the universal fragment of HyperLTL to LTL, and the advantages of the linear translation of LTL to alternating automata, as we now briefly outline. 

\vspace{1ex}
\noindent
      \emph{HyperLTL to LTL.}  Let $\varphi$ be a $\forall^n$-HyperLTL formula and $\cex$ be the counterexample.  
    We construct an LTL formula $\varphi'$ from $\varphi$ as follows~\cite{RVHyper}: atomic propositions indexed with different trace variables are treated as different atomic propositions and trace quantifiers are eliminated. For example $\forall \pi,\pi'. a_\pi \wedge a_{\pi'}$ results in the LTL formula $a_\pi \wedge a_{\pi'}$.
    As for $\cex$, we use the same renaming in order to zip all traces into a single trace, for which we assume the finite representation $t'' = u'' \cdot (v'')^\omega$, which is also the structure of the model checker's output. The trace $t''$ is a violation of the formula $\varphi'$, i.e., $t''$ satisfies $\neg \varphi'$. We denote $\violation := \neg \varphi'$. 
We can then assume, for implementation concerns, that the specification (and its violation) is an LTL formula, and the counterexample is a single trace. 
After our causal analysis, the translation back to a cause over hyperproperties is straightforward as we maintain all information about the different traces in the counterexample. {Note that this translation works due to the synchronous semantics of HyperLTL.}

\vspace{1ex}
\noindent
\emph{Finite Trace Model Checking Using Alternating Automata.} 
In verifying condition CF (that is, in computing counterfactuals and contingencies), we need to apply finite trace model checking, as we want to check if the modified trace in hand still violates the specification $\varphi$, that is, satisfies $\bar\varphi$. To this end, we use the linear algorithm of~\cite{checkingfinitetraces}, that exploits the linear translation of $\violation$ to an alternating automaton $\mathcal{A}_{\violation}$, and using backwards analysis checks the satisfaction of the formula. 
An alternating automaton \cite{AlternatingAutomata} generalizes non-deterministic and universal automata, and its transition relation is a Boolean function over the states. The run of alternating automaton is then a \emph{tree run} that captures the conjunctions in the formula. 
We use the algorithm of~\cite{checkingfinitetraces} as a black box (see App.~\ref{prelim:alt} for a formal definition of alternating automata  
and App.~\ref{app:ltltoalt} for the translation from LTL to alternating automata).
For the computation of contingencies we use an additional feature of the algorithm of~\cite{checkingfinitetraces} -- the algorithm returns an accepting run tree $\mathcal{T}$ of $\mathcal{A}_{\violation}$ on $t''$, with annotations of nodes that represent atomic subformulas of $\violation$ that take part in the satisfaction of $\violation$. We use this feature also in Sec.~\ref{sec:candidates} when calculating the set of candidate causes. 

\subsection{Computing the Set of Candidate Causes}\label{sec:candidates}
The events that might have been a part of the cause to the violation are in fact all events that appear on the counterexample, or, equivalently, all events that appear in $u''$ and $v''$. Note that due to the finite representation, this is a finite set of events. Yet, not all events in this set can cause the violation. In order to remove events that could not have been a part of the cause, we perform an analysis of the transitions of the system taken during the execution of $t''$.
With this analysis we detect which events appearing in the trace locally cause the respective transitions, and thus might be part of the global cause. 
Events that did not trigger a transition in this specific trace cannot be a part of the cause. 
Note that causing a transition and being an actual cause are two different notions - actual causality is defined over the behaviour of the system, not on individual traces.
We denote the over-approximation of the cause as $\candidates$. 
Formally, we represent each transition as a Boolean function over inputs and states. Let $\delta_n$ denote the formula representing the transition of the system taken when reading $t''[n]$, and let $c_{a, n, i}$ be a Boolean variable that corresponds to the event $\langle a_{t_i}, n, t'' \rangle$.\footnote{That is, $\neg c_{a,n,i}$ corresponds to the event $\langle \neg a_{t_i}, n, t'' \rangle$. Recall that the atomic propositions on the zipped trace $t''$ are annotated with the original trace $t_i$ from $\cex$.  } 
Denote $\psi_{n}^t = \bigwedge_{a_{t_i}\in t''[n]} c_{a,n,i} \wedge \bigwedge_{a_{t_i}\notin t''[n]} \neg c_{a,n,i} $, that is, $\psi_n^t$ expresses the exact set of events in $t''[n]$. 
In order to find events that might trigger the transition $\delta_n$, we check
for the \emph{unsatisfiable core} of $ \psi_{n} = (\neg \delta_n) \wedge \psi_n^t$. 
Intuitively, the unsatisfiable core of $\psi_n$ is the set of events that force the system to take this specific transition. For every $c_{a, n, i}$ (or $\neg c_{a, n, i}$ ) in the unsatisfiable core that is also a part of $\psi_n^t$, we add $\langle a, n ,t_i\rangle$ (or $\langle \neg a, n ,t_i\rangle$) to $\candidates$. 

We use unsatisfiable cores in order to find input events that are necessary in order to take a transition. However, this might not be enough. There are cases in which inputs that appear in formula $\violation$ are not detected using this method, as they are not essential in order to take a transition; however, they might be considered a part of the actual cause, as negating them can avoid the violation. Therefore, as a second step, we apply the algorithm of~\cite{checkingfinitetraces}
on the annotated automaton $\mathcal{A}_{\violation}$ in order to find the specific events that affect the satisfaction of $\violation$, and we add these events to $\candidates$. 
Then, the unsatisfiable core approach provides us with inputs that affect the computation and might cause the violation even though they do not appear on the formula itself; while the alternating automaton allows us to find inputs that are not essential for the computation, but might still be a part of the cause as they appear on the formula.

\begin{proposition} \label{prop:overapp}
The set $\candidates$ is indeed an over-approximation of the cause for the violation. That is, every actual cause $\cause$ for the violation is a subset of $\candidates$.
\end{proposition}

\begin{proof}[sketch]
   Let $e= \langle l_a, n, t\rangle $ be an event such that $e$ is not in the unsatisfiable core of $\psi_n$ and does not directly affect the satisfaction of $\violation$ according to the alternating automata analysis. That is, the transition corresponding to $\psi_n^t$ is taken regardless of $e$, and thus all future events on $t$ remain the same regardless of the valuation of $e$. In addition, the valuation of the formula $\violation$ is the same regardless of $e$, since: (1) $e$ does not directly affect the satisfaction of $\violation$; (2) $e$ does not affect future events on $t$ (and obviously it does not affect past events). 
    Therefore, every set $\cause'$ such that $e\in \cause'$ is not minimal, and does not form a cause. Since the above is true for all events $e \not\in \cause$, it holds that $\cause\subseteq \candidates$ for every actual cause $\cause$. \qed
\end{proof}

\begin{proposition} \label{prop:CF}
The set $\candidates$ satisfies conditions SAT and CF. 
\end{proposition}

\begin{proof}
 The condition SAT is satisfied as we add to $\candidates$ only events that indeed occur on the counterexample trace. 
For CF, consider that $\candidates$ is a super-set of the actual cause $\cause$, so the same contingency and counterfactual of $\cause$ will also apply for $\candidates$. This is since in order to compute counterfactual we are allowed to flip any subset of the events in $\cause$, and any such subset is also a subset of $\candidates$. In addition, in computing contingencies, we are allowed to flip any subset of outputs as long as they agree with the counterexample trace, which is independent in $\candidates$ and $\cause$. \qed
\end{proof}

\begin{algorithm}[t]

		\KwIn{Hyperproperty $\varphi$, counterexample $\cex$ violating $\varphi$, and a set of candidate causes $\candidates$ for which conditions SAT and CF hold.}
		\KwOut{A set of input events $\cause$ which is an actual cause for the violation.}
		
		\For{$i \in [1, \ldots, |\candidates|-1]$}{
		\For{$\cause \subset \candidates$ with $|\cause| = i$}{
		let $\cex^f = \intervene(\cex,\cause,\emptyset)$\;  \label{line:intervene}
		\uIf{$\cex^f \models\varphi$ }{
		    \Return{$\cause$\;}}
		    \uElse{
		    $\Tilde{\contingency} = \cont ( \varphi,  \cex, \cause ) $\;
		    	\If{$\Tilde{\contingency} \neq \emptyset$}{
                       \Return{$\cause$\;}
		        }
		    }
		
		}
		}
		\Return{$\candidates$}\;
	
			\caption{$\ActualCause(\varphi, \cex, \candidates)$
			}
		
	\label{alg:causality}
\end{algorithm}	

\vspace{-3ex}

\subsection{Checking Actual Causality}\label{sec:checking_causality}
Due to Prop.~\ref{prop:overapp} we know that in order to find an actual cause, we only need to consider subsets of $\candidates$ as candidate causes. In addition, since $\candidates$ satisfies condition SAT, so do all of its subsets. We thus only need to check conditions CF and MIN for subsets of $\candidates$. 
Our actual causality computation, presented in Alg.~\ref{alg:causality} is as follows. We start with the set $\candidates$, that satisfies SAT and CF. 
We then check if there exists a more minimal cause that satisfies CF. 
This is done by iterating over all subsets $\cause'$ of $\candidates$, ordered by size and starting with the smallest ones, and checking if the counterfactual for the $\cause'$ manages to avoid the violation; and if not, if there exists a contingency for this $\cause'$. 
If the answer to one of these questions is yes, then $\cause'$ is a minimal cause that satisfies SAT, CF, and MIN, and thus we return $\cause'$ as our actual cause. 
We now elaborate on CF and MIN. 

\vspace{1ex}
\noindent
\emph{CF.} As we have mentioned above, checking condition CF is done in two stages -- checking for counterfactuals and computing contingencies. We first show that we do not need to consider all possible counterfactuals, but only one counterfactual for each candidate cause. 

\begin{proposition}
 In order to check if a candidate cause $\candidates$ is an actual cause it is enough to test the one counterfactual where all the events in $\candidates$ are flipped.
\end{proposition}

\begin{proof}
    Assume that there is a strict subset $\cause$ of $\candidates$ such that we only need to flip the valuations of events in $\cause$ in order to find a counterfactual or contingency, thus $\cause$ satisfies CF. Since $\cause$ is a more minimal cause than $\candidates$, we will find it during the minimality check. \qed
\end{proof}

We assume that CF holds for the input set $\candidates$ and check if it holds for any smaller subset $\cause \subset \candidates$. 
CF holds for $\cause$ 
if (1) flipping all events in $\cause$ is enough to avoid the violation of $\varphi$ 
or if (2) there exists a non-empty set of contingencies for $\cause$ that ensures that $\varphi$ is not violated. 
The computation of contingencies is described in Alg.~\ref{alg:contingencies}. Verifying condition CF involves model checking traces against an LTL formula, as we check in Alg.~\ref{alg:causality} (line 3) if the property $\varphi$ is still violated on the counterfactual trace with the empty contingency, and on the counterfactual traces resulting from the different contingency sets we consider in Alg.~\ref{alg:contingencies} (line~7). In both scenarios, we apply finite trace model checking, as described at the beginning of Sec.~\ref{sec:computing-semantic-explanations} (as we assume a lasso-shaped representation of our traces).  

\begin{algorithm}[t]

		\KwIn{Hyperproperty $\varphi$, a counterexample $\cex$ and a potential cause $\cause$.}
		\KwOut{a set of output events $\contingency$ which is a contingency for $\varphi, \cex$ and $\cause$, or $\emptyset$ if no contingency found.}

        \textbf{let} $t''$ be the zipped trace of $\cex$,
        $\varphi'$ be the LTL formula obtained from $\varphi$, and $\violation = \neg \varphi'$\;
        \textbf{let} $\aut{A}_{\violation}$ be the alternating automaton for $\violation$\;
       
     \textbf{let} $t^f$ be the counterfactual trace obtained from $t''$ by flipping all  events in $\cause$\;
  
         \textbf{let} $\mathcal{N}$ be the sets of events derived from the annotated run tree of $\mathcal{A}_{\violation}$ on $t^f$\label{line:annotate}\;

   \textbf{let} $\aut{O}' := \{\langle l_{a_{t}}, n, t''  \rangle \in OE ~|~ a_t \in t''[n] \leftrightarrow  a_t\notin t^f[n]\} $\; \label{line:events}
 
        \For{every subset $\contingency'\subseteq (\aut{N} \cap \aut{O}')$, and then for every other subset $\contingency' \subseteq \aut{O}'$ }{
        
         $t^m := \intervene (t'', \cause, \contingency')$\; 
            \If{$t^m \models \varphi'$}{\Return{$\contingency'$}\; \label{line:annotated_return}}
        }
        
        \Return{$\emptyset$\;}

		\caption{$\cont ( \varphi,  \cex, \cause )$}
		
	\label{alg:contingencies}
\end{algorithm}	

\vspace{1ex}
\noindent
\emph{MIN.} 
To check if $\candidates$ is minimal, we need to check if there exists a subset of $\candidates$ that satisfies CF. 
We check CF for all subsets, starting with the smallest one, and report the first subset that satisfies CF as our actual cause. 
(Note that we already established that $\candidates$ and all of its subsets satisfy SAT.)

\subsection{Computing Contingencies}\label{sec:contingencies}

Recall that the role of contingencies is to eliminate the effect of other possible causes from the counterfactual world, in case these causes did not affect the violation in the actual world. More formally, in computing contingencies we look for a set $\contingency$ of output events such that changing these outputs from their value in the counterfactual to their value in the counterexample $t''$ results in avoiding the violation. Note that the inputs remain as they are in the counterfactual. 
We note that the problem of finding contingencies is hard, and in general is equivalent to the problem of model checking. This is since we need to consider all traces that are the result of changing some subset of events (output + time step) from the counterfactual back to the counterexample, and to check if there exists a trace in this set that avoids the violation.
Unfortunately, we are unable to avoid an exponential complexity in the size of the original system, in the worst case. However, our experiments show that in practice, most cases do not require the use of contingencies. 

Our algorithm for computing contingencies (Alg.~\ref{alg:contingencies}) works as follows. 
Let $t^f$ be the counterfactual trace. 
As a first step, we use the annotated run tree $\aut{T}$ of the alternating automaton $\aut{A}_{\violation}$ on $t^f$ to detect output events that appear in $\violation$ and take part in satisfying $\violation$. Subsets of these output events are our first candidates for contingencies as they are directly related to the violation (Alg.~\ref{alg:contingencies} lines~\ref{line:annotate}-\ref{line:annotated_return}). If we were not able to find a contingency, we continue to check all possible subsets of output events that differ from the original counterexample trace. 
We 
test the different outputs by feeding the \CF\ automaton of Def.~\ref{def:CF_automaton} with additional inputs from the set $I^C$. The resulted trace is then our candidate contingency, which we try to verify against $\varphi$. 
The number of different input sequences is bounded by the size of the product of the \CF\ automaton and the automaton for $\violation$, and thus the process terminates.

\begin{theorem}[Correctness]
    Our algorithm is sound and complete. That is, 
    let $\cex$ be a counterexample with a finite representation to a $\forall^n$-HyperLTL formula $\psi$. Then, our algorithm returns an actual cause for the violation, if such exists. 
\end{theorem}

\begin{proof}
\emph{Soundness}. Since we verify each candidate set of inputs according to the conditions SAT, CF and MIN, it holds that every output of our algorithm is indeed an actual cause.
\emph{Completeness}. If there exists a cause, then due
to Prop.~\ref{prop:overapp}, it is a subset of the finite set $\candidates$. 
Since in the worst case
we test every subset of $\candidates$, if there exists a cause we will eventually find it. 
\qed
\end{proof}

\section{Implementation and Experiments}
\label{sec:experiments}
We implemented Alg.~\ref{alg:causality} and evaluated it on publicly available example instances of \hypervis~\cite{DBLP:journals/tvcg/HorakCMHFMDFD22}, for which their state graphs were available. In the following, we provide implementation details, report on the running time and show the usefulness of the implementation by comparing to the highlighting output of \hypervis. Our implementation is written in Python and uses py-aiger~\cite{py-aiger} and Spot~\cite{DBLP:conf/atva/Duret-LutzLFMRX16}. We compute the candidate cause according to Sec.~\ref{sec:candidates} with py-sat~\cite{imms-sat18}, using Glucose~4~\cite{DBLP:conf/ijcai/AudemardS09,sorensson2010minisat}, building on Minisat~\cite{sorensson2010minisat}.
We ran experiments on a MacBook Pro with a $3,3$ GHz Dual-Core Intel Core i7 processor and $16$ GB RAM\footnote{Our prototype implementation and the experimental data are both available at: \url{https://github.com/reactive-systems/explaining-hyperproperty-violations}}.

\vspace{1ex}
\noindent
\emph{Experimental results.}
\begin{table}[t]
    \caption{Experimental results of our implementation. Times are given in ms.}
	\label{table:results}
	\centering
	\begin{tabular}{*9c}
		\toprule
		\texttt{Instance} & ~$|\cex|$~ & ~$|\varphi|$~ & ~time($\candidates$)~ &  ~$\candidates$~ & ~$\#(\cause)$~ & ~time($\forall\cause$)\\
		\midrule
		\texttt{Running example} & $10$ & $9$ & $19$ & \tiny{$\neg \mathit{hi}^0_{t_1},\mathit{hi}^0_{t_2}$}& $2$ & $55$\\
		\texttt{Security in \& out} & $35$ & $19$ & $292$ &\tiny{$\mathit{hi}^2_{t_1},\neg\mathit{hi}^0_{t_1},\neg\mathit{hi}^3_{t_1},\neg\mathit{hi}^1_{t_1}$}& $8$ & $798$\\
		~  & ~ & ~ & ~ &\tiny{$\mathit{hi}^2_{t_2},\mathit{hi}^0_{t_2},\mathit{hi}^1_{t_2},\mathit{hi}^3_{t_2}$}& ~ & ~\\
		\texttt{Drone example 1} & $24$ & $19$ & $33$ & \tiny{$\mathit{bound}^2_{t_1}, \neg \mathit{bound}^1_{t_1},\mathit{up}^1_{t_1}, \neg\mathit{up}^2_{t_1}$}&$5$&$367$ \\
		~ & ~ & ~ & ~ & \tiny{$\mathit{bound}^2_{t_2}, \neg \mathit{bound}^1_{t_2},\neg \mathit{up}^1_{t_2}$}&~&~ \\
		\texttt{Drone example 2} & $18$ & $36$ & $31$ & \tiny{$\mathit{bound}^1_{t_1},\neg \mathit{bound}^1_{t_2}, \mathit{up}^1_{t_2}$} & $3$& $256$\\
		\texttt{Asymmetric arbiter '19} &  $28$ & $35$ & $53$ & see Appendix~\ref{app:candidates} &$10$ & $490$ \\
		\texttt{Asymmetric arbiter} & $72$ & $35$ & $70$ & see Appendix~\ref{app:candidates} & $24$ & $1480$\\
		\bottomrule
	\end{tabular}
\end{table}
The results of our experimental evaluation can be found in Tab.~\ref{table:results}. We report on the size of the analyzed counterexample $|\cex|$, the size of the violated formula $|\varphi|$, how long it took to compute the first, over-approximated cause (see~time($\candidates$)) and state the approximation $\candidates$ itself, the number of computed minimal causes $\#(\cause)$ and the time it took to compute all of them (see~time($\forall\cause$)). The \texttt{Running Example} is described in Sec.~\ref{sec:running_ex}, the instance \texttt{Security in \& out} refers to a system which leaks high security input by not satisfying a noninterference property, the \texttt{Drone examples} consider a leader-follower drone scenario, and the \texttt{Asymmetric Arbiter} instances refer to arbiter implementations that do not satisfy a symmetry constraint.
Specifications can be found in App.~\ref{app:specs}.

Our first observation is that the cause candidate $\candidates$ can be efficiently computed thanks to the iterative computation of unsatisfiable cores~(Sec.~\ref{sec:candidates}). The cause candidate provides a tight over-approximation of possible minimal causes.
As expected, the runtime for finding minimal causes increases for larger counterexamples. However, as our experiments show, the overhead is manageable, because we optimize the search for all minimal causes by only considering every subset in $\candidates$ instead of naively going over every combination of input events (see Prop.~\ref{prop:overapp}).
Compared to the computationally heavy task of model checking to get a counterexample, our approach incurs little additional cost, which matches our theoretical results~(see Prop.~\ref{cc_as_mc}). During our experiments, we have found that computing the candidate $\candidates$ first has, additionally to providing a powerful heuristic, another benefit: Even when the computation of minimal causes becomes increasingly expensive, $\candidates$ can serve as an intermediate result for the user.
By filtering for important inputs, such as high security inputs, $\candidates$ already gives great insight to why the property was violated.
In the asymmetric arbiter instance, for example, the input events $\langle \neg \mathit{tb\_secret},3,t_0\rangle$ and $\langle \mathit{tb\_secret},3,t_1\rangle$ of $\candidates$, which cause the violation, immediately catch the eye (c.f App.~\ref{app:candidates}).

\vspace{1ex}
\noindent
\emph{Comparison to \hypervis.} \hypervis \cite{DBLP:journals/tvcg/HorakCMHFMDFD22} is a tool for visualizing counterexamples returned from the HyperLTL model checker MCHyper \cite{DBLP:conf/cav/FinkbeinerRS15}. It highlights the events in the trace that it considers responsible for the violation based on the formula and the set of traces, without considering the system model.
However, violations of many relevant security policies such as observational determinism are not caused by events whose atomic propositions appear in the formula, as can be seen in our running example (see Sec.~\ref{sec:running_ex} and Ex.~\ref{ex:hyper_causality}).
When running the highlight function of \hypervis for the counterexample traces $t_1,t_2$ on \texttt{Running example}, the output events $\langle lo, 1, t_1 \rangle$ and $\langle \neg lo,1,t_2\rangle$ will be highlighted, neglecting the decisive high security input $hi$. Using our method additionally reveals the input events $\langle \neg hi, 0, t_1 \rangle$ and $\langle hi,0,t_2\rangle$, i.e., an actual cause (see Tab.~\ref{table:results}). This pattern can be observed throughout all considered instances in our experiments. For instance in the \texttt{Asymmetric arbiter} instance mentioned above, the input events causing the violation also do not occur in the formula (see App.~\ref{app:specs}) and thus \hypervis does not highlight this important cause for the violation.

\section{Related Work}
With the introduction of HyperLTL and HyperCTL$^*$~\cite{DBLP:conf/post/ClarksonFKMRS14}, temporal hyperproperties have been studied extensively:
satisfiability~\cite{conf/concur/FinkbeinerH16,fortin2021hyperltl,mascle2019keys}, model checking~\cite{DBLP:conf/cav/FinkbeinerRS15,hsu2020bounded,finkbeiner2017verifying}, program repair~\cite{DBLP:conf/atva/BonakdarpourF19}, monitoring~\cite{monitoring_hyperproperties_journal,DBLP:conf/csfw/BonakdarpourF18,stucki2019gray,DBLP:conf/csfw/AgrawalB16}, synthesis~\cite{finkbeiner2020synthesis}, and expressiveness studies~\cite{DBLP:conf/lics/CoenenFHH19,krebs2018team,conf/stacs/Finkbeiner017}. Causal analysis of hyperproperties has been studied theoretically based on counterfactual builders~\cite{DBLP:journals/tcs/GosslerS20} instead of actual causality, as in our work. 
Explanation methods~\cite{BaierDFJMPZ21} exist for trace properties~\cite{DBLP:conf/atva/WangYIG06,DBLP:conf/popl/BallNR03,GosslerM13,GroceCKS06,DBLP:conf/cav/GroceKL04}, integrated in several model checkers~\cite{DBLP:conf/tacas/ClarkeKL04,DBLP:journals/tse/ChakiCGJV04,DBLP:conf/sigsoft/ChakiGS04}.
Minimization~\cite{39e69c70fe784b7c99c80df46053a0f4} has been studied, as well as analyzing several system traces together~\cite{DBLP:conf/spin/GroceV03,DBLP:conf/tacas/SchuppanB05,DBLP:conf/hase/BochotVWW10}. 
There exists work in explaining counterexamples for function block diagrams~\cite{DBLP:journals/acj/JeeJCKYPS10,DBLP:conf/indin/PakonenBV18}.
MODCHK uses a causality analysis~\cite{DBLP:conf/cav/BeerBCOT09} returning an over-approximation, while we provide minimal causes.
Lastly, there are approaches which define actual causes for the violation of a trace property  using Event Order Logic~\cite{CaltaisGL18,Leitner-FischerL13a,Leitner-FischerL13b}.

\section{Conclusion}
We present an explanation method for counterexamples to hyperproperties described by HyperLTL formulas. We lift Halpern and Pearl's definition of actual causality to effects described by hyperproperties and counterexamples given as sets of traces. Like the definition that inspired us, we allow modifications of the system dynamics in the counterfactual world through contingencies, and define these possible counterfactual behaviors in an automata-theoretic approach. The evaluation of our prototype implementation shows that our method is practically applicable and significantly improves the state-of-the-art in explaining counterexamples returned by a HyperLTL model checker.

\bibliographystyle{splncs04}
\bibliography{bibliography}

\begin{thebibliography}{10}
\providecommand{\url}[1]{\texttt{#1}}
\providecommand{\urlprefix}{URL }
\providecommand{\doi}[1]{https://doi.org/#1}

\bibitem{log4j}
Apache log4j security vulnerabilities,
  \url{https://logging.apache.org/log4j/2.x/security.html}

\bibitem{DBLP:conf/csfw/AgrawalB16}
Agrawal, S., Bonakdarpour, B.: Runtime verification of k-safety hyperproperties
  in hyperltl. In: {CSF} 2016. \doi{10.1109/CSF.2016.24}

\bibitem{DBLP:conf/ijcai/AudemardS09}
Audemard, G., Simon, L.: Predicting learnt clauses quality in modern {SAT}
  solvers. In: {IJCAI} 2009.
  \url{http://ijcai.org/Proceedings/09/Papers/074.pdf}

\bibitem{BaierDFJMPZ21}
Baier, C., Dubslaff, C., Funke, F., Jantsch, S., Majumdar, R., Piribauer, J.,
  Ziemek, R.: {From Verification to Causality-Based Explications}. In: ICALP
  2021

\bibitem{DBLP:conf/popl/BallNR03}
Ball, T., Naik, M., Rajamani, S.K.: From symptom to cause: localizing errors in
  counterexample traces. In: POPL 2003. \doi{10.1145/604131.604140}

\bibitem{DBLP:journals/mscs/BartheDR11}
Barthe, G., D'Argenio, P.R., Rezk, T.: Secure information flow by
  self-composition. Math. Struct. Comput. Sci.  \textbf{21}(6),  1207--1252
  (2011)

\bibitem{DBLP:conf/cav/BeerBCOT09}
Beer, I., Ben{-}David, S., Chockler, H., Orni, A., Trefler, R.J.: Explaining
  counterexamples using causality. In: Computer Aided Verification, {CAV} 2009

\bibitem{Biere-FMV-TR-07-1}
Biere, A.: The {AIGER And-Inverter Graph (AIG)} format version 20071012. Tech.
  Rep. Report 07/1, Johannes Kepler University (2007)

\bibitem{DBLP:conf/hase/BochotVWW10}
Bochot, T., Virelizier, P., Waeselynck, H., Wiels, V.: Paths to property
  violation: {A} structural approach for analyzing counter-examples. In: {HASE}
  2010

\bibitem{DBLP:conf/csfw/BonakdarpourF18}
Bonakdarpour, B., Finkbeiner, B.: The complexity of monitoring hyperproperties.
  In: {CSF} 2018. \doi{10.1109/CSF.2018.00019}

\bibitem{DBLP:conf/atva/BonakdarpourF19}
Bonakdarpour, B., Finkbeiner, B.: Program repair for hyperproperties. In:
  {ATVA} 2019. \doi{10.1007/978-3-030-31784-3\_25}

\bibitem{10.1007/978-3-642-14295-6_5}
Brayton, R., Mishchenko, A.: Abc: An academic industrial-strength verification
  tool. In: {CAV} 2010. \doi{10.1007/978-3-642-14295-6\_5}

\bibitem{CaltaisGL18}
Caltais, G., Guetlein, S.L., Leue, S.: Causality for general ltl-definable
  properties. In: CREST@ETAPS 2018. \doi{10.4204/EPTCS.286.1}

\bibitem{DBLP:journals/tse/ChakiCGJV04}
Chaki, S., Clarke, E.M., Groce, A., Jha, S., Veith, H.: Modular verification of
  software components in {C}. {IEEE} Trans. Software Eng.  \textbf{30}(6),
  388--402 (2004)

\bibitem{DBLP:conf/sigsoft/ChakiGS04}
Chaki, S., Groce, A., Strichman, O.: Explaining abstract counterexamples. In:
  {ACM} {SIGSOFT} Fnd. of Soft. Eng. 2004. \doi{10.1145/1029894.1029908}

\bibitem{DBLP:journals/tocl/ChocklerHK08}
Chockler, H., Halpern, J.Y., Kupferman, O.: What causes a system to satisfy a
  specification? {ACM} Trans. Comput. Log.  \textbf{9}(3),  20:1--20:26 (2008)

\bibitem{DBLP:journals/fmsd/ClarkeBRZ01}
Clarke, E.M., Biere, A., Raimi, R., Zhu, Y.: Bounded model checking using
  satisfiability solving. Formal Methods Syst. Des.  \textbf{19}(1),  7--34
  (2001)

\bibitem{DBLP:conf/lop/ClarkeE81}
Clarke, E.M., Emerson, E.A.: Design and synthesis of synchronization skeletons
  using branching-time temporal logic. In: Logics of Programs, Workshop,
  Yorktown Heights, New York, USA, May 1981. \doi{10.1007/BFb0025774}

\bibitem{DBLP:conf/tacas/ClarkeKL04}
Clarke, E.M., Kroening, D., Lerda, F.: A tool for checking {ANSI-C} programs.
  In: {TACAS} 2004. \doi{10.1007/978-3-540-24730-2\_15}

\bibitem{DBLP:conf/post/ClarksonFKMRS14}
Clarkson, M.R., Finkbeiner, B., Koleini, M., Micinski, K.K., Rabe, M.N.,
  S{\'{a}}nchez, C.: Temporal logics for hyperproperties. In: {POST} 2014

\bibitem{Hyperproperties}
Clarkson, M.R., Schneider, F.B.: Hyperproperties. J. Comput. Secur.
  \textbf{18}(6),  1157--1210 (2010). \doi{10.3233/JCS-2009-0393}

\bibitem{DBLP:conf/lics/CoenenFHH19}
Coenen, N., Finkbeiner, B., Hahn, C., Hofmann, J.: The hierarchy of
  hyperlogics. In: {LICS} 2019. \doi{10.1109/LICS.2019.8785713}

\bibitem{DBLP:conf/cav/CoenenFST19}
Coenen, N., Finkbeiner, B., S{\'{a}}nchez, C., Tentrup, L.: Verifying
  hyperliveness. In: {CAV} 2019. \doi{10.1007/978-3-030-25540-4\_7}

\bibitem{softwaredoping}
D'Argenio, P.R., Barthe, G., Biewer, S., Finkbeiner, B., Hermanns, H.: Is your
  software on dope? In: ESOP 2017

\bibitem{DBLP:conf/cav/DehnertJK017}
Dehnert, C., Junges, S., Katoen, J., Volk, M.: A storm is coming: {A} modern
  probabilistic model checker. In: Computer Aided Verification, {CAV} 2017

\bibitem{DBLP:conf/atva/Duret-LutzLFMRX16}
Duret{-}Lutz, A., Lewkowicz, A., Fauchille, A., Michaud, T., Renault, E., Xu,
  L.: Spot 2.0 --- a framework for ltl and $\omega $-automata manipulation. In:
  {ATVA} 2016

\bibitem{durumeric2014matter}
Durumeric, Z., Li, F., Kasten, J., Amann, J., Beekman, J., Payer, M., Weaver,
  N., Adrian, D., Paxson, V., Bailey, M., et~al.: The matter of heartbleed. In:
  {IMC} 2014

\bibitem{conf/concur/FinkbeinerH16}
Finkbeiner, B., Hahn, C.: Deciding hyperproperties. In: {CONCUR} 2016

\bibitem{finkbeiner2020synthesis}
Finkbeiner, B., Hahn, C., Lukert, P., Stenger, M., Tentrup, L.: Synthesis from
  hyperproperties. Acta Informatica  \textbf{57}(1-2),  137--163 (2020)

\bibitem{RVHyper}
Finkbeiner, B., Hahn, C., Stenger, M., Tentrup, L.: Rvhyper: {A} runtime
  verification tool for temporal hyperproperties. In: {TACAS} 2018

\bibitem{monitoring_hyperproperties_journal}
Finkbeiner, B., Hahn, C., Stenger, M., Tentrup, L.: Monitoring hyperproperties.
  Formal Methods Syst. Des.  \textbf{54}(3),  336--363 (2019)

\bibitem{finkbeiner2018model}
Finkbeiner, B., Hahn, C., Torfah, H.: Model checking quantitative
  hyperproperties. In: {CAV} 2018. \doi{10.1007/978-3-319-96145-3\_8}

\bibitem{finkbeiner2017verifying}
Finkbeiner, B., M{\"{u}}ller, C., Seidl, H., Zalinescu, E.: Verifying security
  policies in multi-agent workflows with loops. In: {CCS} 2017

\bibitem{DBLP:conf/cav/FinkbeinerRS15}
Finkbeiner, B., Rabe, M.N., S{\'{a}}nchez, C.: Algorithms for model checking
  hyperltl and hyperctl$^*$. In: {CAV} 2015. \doi{10.1007/978-3-319-21690-4\_3}

\bibitem{checkingfinitetraces}
Finkbeiner, B., Sipma, H.: Checking finite traces using alternating automata.
  Formal Methods Syst. Des.  \textbf{24}(2),  101--127 (2004)

\bibitem{conf/stacs/Finkbeiner017}
Finkbeiner, B., Zimmermann, M.: The first-order logic of hyperproperties. In:
  {STACS} 2017. \doi{10.4230/LIPIcs.STACS.2017.30}

\bibitem{fortin2021hyperltl}
Fortin, M., Kuijer, L.B., Totzke, P., Zimmermann, M.: {HyperLTL} satisfiability
  is {$\Sigma_1^1$}-complete, {HyperCTL$^*$} satisfiability is
  {$\Sigma_1^2$}-complete. In: {MFCS} 2021

\bibitem{GosslerM13}
G{\"{o}}{\ss}ler, G., M{\'{e}}tayer, D.L.: A general trace-based framework of
  logical causality. In: {FACS} 2013. \doi{10.1007/978-3-319-07602-7\_11}

\bibitem{DBLP:journals/tcs/GosslerS20}
G{\"{o}}ssler, G., Stefani, J.: Causality analysis and fault ascription in
  component-based systems. Theor. Comput. Sci.  \textbf{837},  158--180 (2020)

\bibitem{GroceCKS06}
Groce, A., Chaki, S., Kroening, D., Strichman, O.: Error explanation with
  distance metrics. Int. J. Softw. Tools Technol. Transf.  \textbf{8}(3),
  229--247 (2006)

\bibitem{DBLP:conf/cav/GroceKL04}
Groce, A., Kroening, D., Lerda, F.: Understanding counterexamples with explain.
  In: {CAV} 2004. \doi{10.1007/978-3-540-27813-9\_35}

\bibitem{DBLP:conf/spin/GroceV03}
Groce, A., Visser, W.: What went wrong: Explaining counterexamples. In: {SPIN}
  2003. \doi{10.1007/3-540-44829-2\_8}

\bibitem{Halpern15}
Halpern, J.Y.: A modification of the halpern-pearl definition of causality. In:
  {IJCAI} 2015. \url{http://ijcai.org/Abstract/15/427}

\bibitem{HalpernPearl05a}
Halpern, J.Y., Pearl, J.: Causes and explanations: A structural-model approach.
  part i: Causes. The British Journal for the Philosophy of Science
  \textbf{56}(4),  843--887 (2005), \url{http://www.jstor.org/stable/3541870}

\bibitem{HalpernPearl05b}
Halpern, J.Y., Pearl, J.: Causes and explanations: A structural-model approach.
  part ii: Explanations. The British Journal for the Philosophy of Science
  \textbf{56}(4),  889--911 (2005), \url{http://www.jstor.org/stable/3541871}

\bibitem{DBLP:journals/tse/Holzmann97}
Holzmann, G.J.: The model checker {SPIN}. {IEEE} Trans. Software Eng.
  \textbf{23}(5),  279--295 (1997). \doi{10.1109/32.588521}

\bibitem{DBLP:journals/tvcg/HorakCMHFMDFD22}
Horak, T., Coenen, N., Metzger, N., Hahn, C., Flemisch, T., M{\'{e}}ndez, J.,
  Dimov, D., Finkbeiner, B., Dachselt, R.: Visual analysis of hyperproperties
  for understanding model checking results. {IEEE} Trans. Vis. Comput. Graph.
  \textbf{28}(1),  357--367 (2022)

\bibitem{hsu2020bounded}
Hsu, T., S{\'{a}}nchez, C., Bonakdarpour, B.: Bounded model checking for
  hyperproperties. In: {TACAS} 2021. \doi{10.1007/978-3-030-72016-2\_6}

\bibitem{imms-sat18}
Ignatiev, A., Morgado, A., Marques{-}Silva, J.: {PySAT:} {A} {Python} toolkit
  for prototyping with {SAT} oracles. In: SAT. pp. 428--437 (2018)

\bibitem{DBLP:journals/acj/JeeJCKYPS10}
Jee, E., Jeon, S., Cha, S.D., Koh, K.Y., Yoo, J., Park, G., Seong, P.:
  Fbdverifier: Interactive and visual analysis of counterexample in formal
  verification of function block diagram. J. Res. Pract. Inf. Technol.
  \textbf{42}(3),  171--188 (2010)

\bibitem{Kocher2018spectre}
Kocher, P., Horn, J., Fogh, A., , Genkin, D., Gruss, D., Haas, W., Hamburg, M.,
  Lipp, M., Mangard, S., Prescher, T., Schwarz, M., Yarom, Y.: Spectre attacks:
  Exploiting speculative execution. In: SP 2019. \doi{10.1109/SP.2019.00002}

\bibitem{krebs2018team}
Krebs, A., Meier, A., Virtema, J., Zimmermann, M.: Team semantics for the
  specification and verification of hyperproperties. In: {MFCS} 2018

\bibitem{39e69c70fe784b7c99c80df46053a0f4}
Lahtinen, J., Launiainen, T., Heljanko, K., {Ropponen}, J.: Model checking
  methodology for large systems, faults and asynchronous behaviour: SARANA 2011
  work report. No.~12 in VTT Tech., VTT Tech. Research Centre of Finland (2012)

\bibitem{DBLP:journals/sttt/LarsenPY97}
Larsen, K.G., Pettersson, P., Yi, W.: {UPPAAL} in a nutshell. Int. J. Softw.
  Tools Technol. Transf.  \textbf{1}(1-2) (1997),
  \url{https://doi.org/10.1007/s100090050010}

\bibitem{Leitner-FischerL13a}
Leitner{-}Fischer, F., Leue, S.: Causality checking for complex system models.
  In: {VMCAI} 2013. \doi{10.1007/978-3-642-35873-9\_16}

\bibitem{Leitner-FischerL13b}
Leitner{-}Fischer, F., Leue, S.: Probabilistic fault tree synthesis using
  causality computation. Int. J. Crit. Comput. Based Syst.  \textbf{4}(2),
  119--143 (2013)

\bibitem{Lewis1973}
Lewis, D.: Causation. Journal of Philosophy  \textbf{70}(17),  556--567 (1973)

\bibitem{Lipp2018meltdown}
Lipp, M., Schwarz, M., Gruss, D., Prescher, T., Haas, W., Horn, J., Mangard,
  S., Kocher, P., Genkin, D., Yarom, Y., Hamburg, M., Strackx, R.: Meltdown:
  reading kernel memory from user space. Commun. {ACM}  \textbf{63}(6),  46--56
  (2020)

\bibitem{mascle2019keys}
Mascle, C., Zimmermann, M.: The keys to decidable hyperltl satisfiability:
  Small models or very simple formulas. In: {CSL} 2020

\bibitem{GNI}
McCullough, D.: Noninterference and the composability of security properties.
  In: Proceedings. 1988 IEEE Symposium on Security and Privacy. pp. 177--186
  (1988)

\bibitem{moore1956gedanken}
Moore, E.F.: Gedanken-experiments on sequential machines. Aut. stud.
  \textbf{34} (1956)

\bibitem{DBLP:conf/indin/PakonenBV18}
Pakonen, A., Buzhinsky, I., Vyatkin, V.: Counterexample visualization and
  explanation for function block diagrams. In: {INDIN} 2018

\bibitem{DBLP:conf/focs/Pnueli77}
Pnueli, A.: The temporal logic of programs. In: {FOCS} 1977

\bibitem{DBLP:conf/tacas/SchuppanB05}
Schuppan, V., Biere, A.: Shortest counterexamples for symbolic model checking
  of {LTL} with past. In: {TACAS} 2005. \doi{10.1007/978-3-540-31980-1\_32}

\bibitem{sorensson2010minisat}
S{\"o}rensson, N.: Minisat 2.2 and minisat++ 1.1. SAT Race 2010

\bibitem{stucki2019gray}
Stucki, S., S{\'{a}}nchez, C., Schneider, G., Bonakdarpour, B.: Gray-box
  monitoring of hyperproperties. In: {FM} 2019

\bibitem{AlternatingAutomata}
Vardi, M.Y.: Alternating automata: Unifying truth and validity checking for
  temporal logics. In: CADE-14 (1997). \doi{10.1007/3-540-63104-6\_19}

\bibitem{py-aiger}
Vazquez-C., M., Rabe, M.: {py-aiger},
  \url{https://github.com/mvcisback/py-aiger}

\bibitem{DBLP:conf/atva/WangYIG06}
Wang, C., Yang, Z., Ivancic, F., Gupta, A.: Whodunit? causal analysis for
  counterexamples. In: {ATVA} 2006. \doi{10.1007/11901914\_9}

\end{thebibliography}

\appendix

\section{Appendix}
\label{app:additional:definitions}
\subsection{Counterfactual Automaton for Running Example}\label{app:counterfactual_automaton}

The counterfactual automaton $T^C_{t_2}$ for system T (see Fig.~\ref{fig:running_example_system}) and the trace $t_2 = \{\mathit{hi}\} \{\mathit{hi},\mathit{ho}\} \{\mathit{ho},\mathit{lo}\}^\omega$ from the running example (see Sec.~\ref{sec:running_ex}) is illustrated in Fig.~\ref{fig:counterfactual_automaton}. For readability we only pictured the reachable fragment of the automaton, the states $(s_1,0),(s_2,0)$ and $(s_3,0)$ are missing. Note that following $t_2$ can be done in $T^C_{t_2}$ without the contingency variables, e.g., $\mathit{lo}^C$. If we diverge from the input sequence of $t_2$, however, such as when taking an edge labeled with $\lnot \mathit{hi}$ in the first step, then we can use the contingency variables to set the corresponding value as in $t_2$. For instance, the edge labeled with $\lnot \mathit{hi} \land \mathit{lo}^C \land \lnot \mathit{ho}^C$ goes to $(s_0,1)$ because we remove $\mathit{lo}$ from the outputs of the successor state. We then proceed as in the state with the corresponding label, i.e., $s_0$. Automata copies in the chain that correspond to the prefix of the trace cannot be visited twice, for instance in $T^C_{t_2}$ we are forced to loop in the third copy and can visit the first two only once.

\begin{figure}
    \centering
    \hspace*{-1cm}
\begin{tikzpicture}
			[->,shorten >=1pt,auto,node distance=3.5cm and 3cm, on grid,initial text=,
			every state/.style={minimum size=45pt,inner sep=0pt},
			every node/.style={font=\small}]
			
			\node(s0) [state, initial above] {\begin{tabular}{c}
					$s_0,0$\\ $\emptyset$
				\end{tabular}};

            \node(s01) [state, below left of=s0,xshift=-7.5em,yshift=-6em] {\begin{tabular}{c}
					$s_0,1$\\ $\emptyset$
				\end{tabular}};
			\node(s11) [state, right of=s01] {\begin{tabular}{c}
					$s_1.1$\\ $\{\mathit{ho}\}$
			\end{tabular}};
		    \node(s21) [state, right of=s11] {\begin{tabular}{c}
		    $s_2.1$\\ $\{\mathit{lo}\}$
	        \end{tabular}};
            \node(s31) [state, right of=s21] {\begin{tabular}{c}
            $s_3,1$\\ $\{\mathit{ho},\mathit{lo}\}$
            \end{tabular}};
            
            \node(s02) [state, below of=s01,yshift=-10em] {\begin{tabular}{c}
					$s_0,2$\\ $\emptyset$
				\end{tabular}};
			\node(s12) [state, right of=s02] {\begin{tabular}{c}
					$s_1.2$\\ $\{\mathit{ho}\}$
			\end{tabular}};
		    \node(s22) [state, right of=s12] {\begin{tabular}{c}
		    $s_2.2$\\ $\{\mathit{lo}\}$
	        \end{tabular}};
            \node(s32) [state, right of=s22] {\begin{tabular}{c}
            $s_3,2$\\ $\{\mathit{ho},\mathit{lo}\}$
            \end{tabular}};
			
			\path[->]
			(s0)	edge[] node[swap]{$\mathit{hi}$}  (s11)
			(s0)	edge[bend right] node[swap,pos=0.25]{$\lnot \mathit{hi} \land \mathit{lo}^C \land \lnot \mathit{ho}^C$}  (s01)
			(s0)	edge[bend left] node[pos=0.25]{$\lnot \mathit{hi} \land \mathit{ho}^C$}  (s31)
			(s0)	edge[] node[]{$\lnot \mathit{hi} \land \lnot \mathit{lo}^C \land \lnot \mathit{ho}^C$}  (s21)
			(s31)	edge[] node[]{$\top$}  (s32)
			(s11)	edge[] node[]{$\top$}  (s32)
			(s21)	edge node[pos=0.1,swap,xshift=0.5em]{$\mathit{hi} \land \lnot \mathit{lo}^C$}  (s12)
			(s21)	edge node[pos=0.1]{$\lnot \mathit{hi} \lor \mathit{lo}^C$}  (s32)
			(s12)	edge[bend right] node[swap]{$\top$} (s32)
			(s32)	edge[loop below] node[]{$\top$} (s32)
			(s22)	edge[] node[]{$\lnot \mathit{hi} \lor \mathit{lo}^C$} (s32)
			(s22)	edge[] node[swap]{$\mathit{hi} \land \lnot \mathit{lo}^C$} (s12)
			(s02)	edge[bend right=45] node[swap]{$\mathit{ho}^C \land \mathit{lo}^C$} (s32)
			(s02)	edge[bend right] node[swap]{$\lnot \mathit{hi} \land \lnot \mathit{ho}^C \land \lnot \mathit{lo}^C$} (s22)
			(s02)	edge[bend left] node[yshift=0.25em,xshift=-4em]{$\mathit{hi} \land \lnot \mathit{ho}^C \land \lnot \mathit{lo}^C$} (s12)
			(s01)	edge[] node[pos=0.2]{$\mathit{ho}^C \land \mathit{lo}^C$} (s32)
			(s01)	edge[] node[swap]{$\lnot \mathit{hi} \land \lnot \mathit{ho}^C \land \lnot \mathit{lo}^C$} (s22)
			(s01)	edge[bend right] node[swap]{$\mathit{hi} \land \lnot \mathit{ho}^C \land \lnot \mathit{lo}^C$} (s12);
\end{tikzpicture}
    \caption{The counterfactual automaton $T^C_{t_2}$ described in App.~\ref{app:counterfactual_automaton}.}
    \label{fig:counterfactual_automaton}
\end{figure}
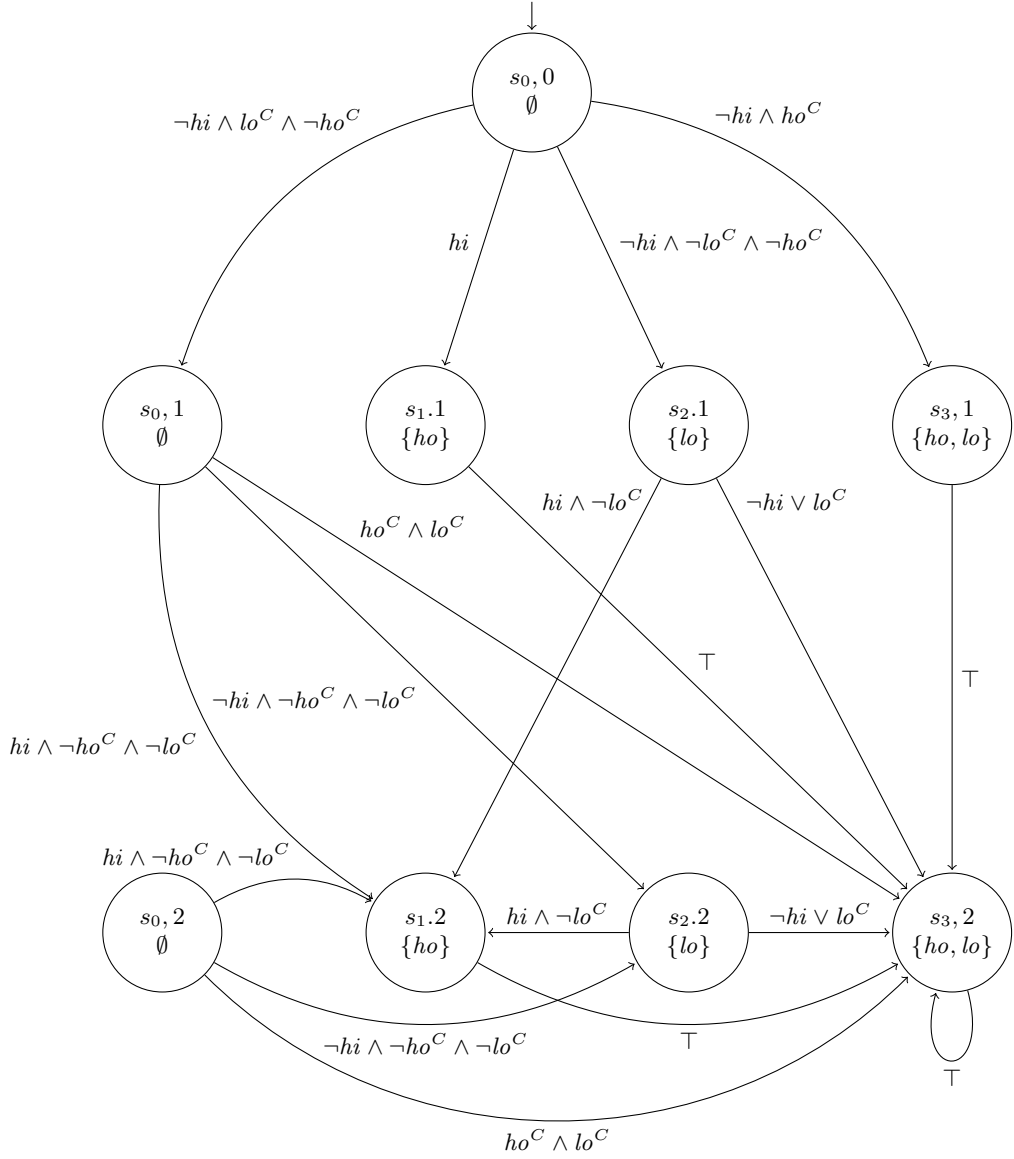

\subsection{Alternating Automata}\label{prelim:alt} 
Alternating automata \cite{AlternatingAutomata} are automata over infinite words that generalize non-deterministic and universal automata.
An alternating automaton $\mathcal{A}$ is defined by the following grammar $$\epsilon_\mathcal{A} \mid \langle \nu, \mathcal{A}, f\rangle \mid \mathcal{A} \wedge \mathcal{A} \mid \mathcal{A} \vee \mathcal{A}$$ where $\epsilon_\mathcal{A}$ is the empty automaton;  $\mathcal{A} \wedge \mathcal{A}$ and $\mathcal{A} \vee \mathcal{A}$ are a disjunction and conjunction of two automata, respectively; and 
$n = \langle \nu, \mathcal{A}, f\rangle$ is a node, such that 
$\nu$ is a state formula that labels $n$, $\mathcal{A}$  is the next state (automaton), and $f$ indicates whether $n$ is accepting or rejecting (\emph{acc}/\emph{rec}).
Since runs of an alternating automaton are defined using conjunctions, they form a \emph{run tree} (see~\cite{AlternatingAutomata} for a formal definition), where disjunctions express non-determinism. 
The set of words accepted (using Büchi acceptance condition) by an alternating automaton are those who have a run such that all of its branches in the tree visit infinitely often in an accepting state. The language $\mathcal{L}(\mathcal{A})$ of an automaton $\mathcal{A}$ is the set of words accepted by $\mathcal{A}$. 

The set of nodes of an alternating automaton $\mathcal{A}$ is denoted by $\mathcal{N}(\mathcal{A})$, where $\mathcal{N}(\epsilon_{\mathcal{A}}) = \emptyset$, $\mathcal{N}(\langle \nu, \mathcal{A}, f\rangle) = \mathcal{N}(\mathcal{A})$, and $\mathcal{N}(\mathcal{A} \wedge\mathcal{A}')= \mathcal{N}(\mathcal{A} \vee \mathcal{A}')= \mathcal{N}(\mathcal{A}) \cup \mathcal{N}(\mathcal{A}')$.
For an LTL formula $\varphi$ there is a linear translation to an  
alternating automaton $\mathcal{A_\varphi}$, s.t. $\mathcal{L}(\mathcal{A_\varphi})$ is the set of traces that satisfy $\varphi$~\cite{AlternatingAutomata}.
In this construction, $\mathcal{N}(\mathcal{A})$ is the set of subformulas of $\varphi$ and their negations, $\mathcal{A}$ is $\varphi$, and all formulas of the form $\neg (\varphi_1 \LTLuntil \varphi_2)$ are \emph{accepting}.

\subsection{LTL to Alternating Automata}\label{app:ltltoalt}
For LTL formulas $\varphi$ and $\varphi'$ and $a \in AP$
\begin{align*}
    \alt(a) &= \langle a, \epsilon_\alt, acc \rangle\\
    \alt(\neg a) &= \langle \neg a, \epsilon_\alt, acc \rangle\\
    \alt(\varphi \wedge \varphi')  &= \alt(\varphi) \wedge \alt(\varphi')\\
    \alt(\varphi \vee \varphi')  &= \alt(\varphi) \vee \alt(\varphi')\\
    \alt(\LTLnext \varphi) &= \langle \mathit{true}, \alt(\varphi), \mathit{rej}\rangle\\
    \alt(\LTLglobally \varphi) &= \langle \mathit{true}, \alt(\LTLglobally \varphi), \mathit{acc}\rangle \wedge \alt(\varphi)\\
    \alt(\LTLeventually \varphi) &= \langle \mathit{true}, \alt(\LTLeventually \varphi), \mathit{rej}\rangle \wedge \alt(\varphi)\\
    \alt(\varphi \LTLuntil \varphi') &=\alt(\varphi') \vee (\langle \mathit{true}, \alt(\varphi \LTLuntil \varphi'), \mathit{rej}\rangle \wedge \alt(\varphi))\\
    \alt(\varphi \LTLrelease \varphi') &= (\alt(\varphi)\wedge \alt(\varphi')) \vee (\langle \mathit{true}, \alt(\varphi \LTLrelease \varphi'), \mathit{acc}\rangle \wedge \alt(\varphi'))
\end{align*}

\subsection{Cause Candidates for Asymmetric Arbiter}
\label{app:candidates}
In the following, we provide the cause candidates $\candidates$, which are given as an intermediate output of our implementation for both asymmetric arbiter instances.
Especially in the \texttt{Asymmetric arbiter} instance, that $tb\_secret$ does not hold on timestep $3$ on trace $t_1$, but does hold on timestep $3$ on trace $t_2$ immediately catches the eye. The events are given as a tuple of the valuation of an input and its timestep, sorted by their trace belonging.

\lstset{%
  breaklines=true,
  breakatwhitespace=true,
}
\noindent
$\candidates$ for \texttt{Asymmetric arbiter}:
\begin{lstlisting}
T1: [(!req_1, 6), (!req_0, 1), (req_1, 0), (req_0, 0), (req_0, 5), (req_1, 4), (!req_1, 2), (!req_0, 2), (!tb_secret, 3), (req_0, 3), (req_1', 3)]
T2: [(!req_0, 4), (!req_0, 1), (req_1, 0), (req_0, 0), (!req_0, 6), (!req_1, 6), (req_0, 5), (req_1, 5), (req_0, 3), (tb_secret, 3), (req_1, 3), (!req_1, 2), (!req_0, 2)]
\end{lstlisting}

\noindent
$\candidates$ for \texttt{Asymmetric arbiter '19}:
\begin{lstlisting}
T1: [(!req0, 1), (req0, 0), (req1, 0), (!req0, 2), (!req1, 2)]
T2: [(!req0', 1), (req0, 0), (req1, 0), (!req0, 2), (!req1, 2)]
\end{lstlisting}

\subsection{Specifications for Experimental Results}
\label{app:specs}
We used the following HyperLTL specifications given in MCHyper syntax on the respective benchmarks.
\lstset{%
  breaklines=true,
  breakatwhitespace=true,
}

\noindent
\texttt{Running Example}:
\begin{lstlisting}
Forall (Forall (G (Eq (AP \"lo\" 0) (AP \"lo\" 1))))
\end{lstlisting}
\texttt{Security in \& out}:
\begin{lstlisting}
Forall (Forall (Implies (G (Eq (AP \"li\" 0) (AP \"li\" 1))) (G (Eq (AP \"lo\" 0) (AP \"lo\" 1)))))
\end{lstlisting}
\texttt{Drone example 1}:
\begin{lstlisting}
Forall (Forall (X (G (Implies (Eq (AP \"bound\" 0) (AP \"bound\" 1)) (X (Eq (AP \"emergency\" 0) (AP \"emergency\" 1)))))))
\end{lstlisting}
\texttt{Drone example 2}:
\begin{lstlisting}
Forall (Forall (Or (G (Eq (AP \"increase\" 0) (AP \"increase\" 1))) (Until (Eq (AP \"increase\" 0) (AP \"increase\" 1)) (And (Eq (AP \"increase\" 0) (AP \"increase\" 1)) (Neq (AP \"up\" 0) (AP \"up\" 1))))))
\end{lstlisting}
\texttt{Asymmetric arbiter '19}:
\begin{lstlisting}
Forall (Forall (Implies (G (And (Eq (AP \"req0\" 0) (AP \"req1\" 1)) (Eq (AP \"req1\" 0) (AP \"req0\" 1)))) (G (And (Eq (AP \"grant0\" 0) (AP \"grant1\" 1)) (Eq (AP \"grant1\" 0) (AP \"grant0\" 1))))))
\end{lstlisting}
\texttt{Asymmetric arbiter}:
\begin{lstlisting}
Forall (Forall (Implies (G (And (Eq (AP \"req_0\" 0) (AP \"req_0\" 1)) (Eq (AP \"req_1\" 0) (AP \"req_1\" 1)))) (G (And (Eq (AP \"grant_0\" 0) (AP \"grant_0\" 1)) (Eq (AP \"grant_1\" 0) (AP \"grant_1\" 1))))))
\end{lstlisting}

\end{document}